\documentclass[journal]{IEEEtran}

\usepackage[utf8]{inputenc} 
\usepackage[T1]{fontenc}
\usepackage{graphicx,colordvi,psfrag}
\usepackage{amsmath,amssymb}
\usepackage{epstopdf}
\usepackage{epsfig}
\usepackage{calc,pstricks, pgf, xcolor}
\usepackage{bbding}
\usepackage{bbm}
\usepackage{dsfont}
\usepackage{centernot}
\usepackage[caption=false]{subfig}
\usepackage{setspace}
\usepackage{ifthen}
\usepackage{xcolor}

\newcommand{\Ind}{\mathds{1}}

\newcommand{\by}{\mathbf{y}}
\newcommand{\bY}{\mathbf{Y}}
\newcommand{\bx}{\mathbf{x}}
\newcommand{\bX}{\mathbf{X}}

\newcommand{\bZ}{\mathbf{Z}}

\newcommand{\bH}{\mathbf{H}}
\newcommand{\bA}{\mathbf{A}}

\newcommand{\ba}{\mathbf{a}}
\newcommand{\bc}{\mathbf{c}}

\newcommand{\bD}{\mathbf{D}}
\newcommand{\bB}{\mathbf{B}}
\newcommand{\bb}{\mathbf{b}}
\newcommand{\bU}{\mathbf{U}}
\newcommand{\bu}{\mathbf{u}}
\newcommand{\bV}{\mathbf{V}}
\newcommand{\bv}{\mathbf{v}}

\newcommand{\bw}{\mathbf{w}}
\newcommand{\bt}{\mathbf{t}}

\newcommand{\be}{\mathbf{e}}

\newcommand{\bR}{\mathbf{R}}
\newcommand{\bI}{\mathbf{I}}

\newcommand{\bG}{\mathbf{G}}
\newcommand{\bQ}{\mathbf{Q}}

\newcommand{\bS}{\mathbf{S}}

\newcommand{\ZZ}{\mathbb{Z}}
\newcommand{\RR}{\mathbb{R}}

\newcommand{\diag}{\mathop{\mathrm{diag}}}

\newcommand{\rank}{\mathop{\mathrm{rank}}}

\newcommand{\Vol}{\mathrm{Vol}}

\newcommand{\CUBE}{\mathrm{CUBE}}
\newcommand{\SLk}{\mathrm{GL}_K(\ZZ)}
\newcommand{\bSigma}{\mathbf{\Sigma}}
\newcommand{\bSigmaV}{{\overset{\vee}{\bSigma}}}

\newcommand{\m}{\mathcal}

\newcommand{\abs}[1]{\left \vert #1 \right \vert}
\newcommand{\norm}[1]{\left \Vert #1 \right \Vert}
\newcommand{\E}[1]{\mathbb{E} \left[ #1 \right]}
\newcommand{\Bin}{\mathrm{Binomial}}
\newcommand{\Cov}{\mathrm{Cov}}

\makeatletter
\newcommand{\oset}[2]{%
	{\mathop{#2}\limits^{\vbox to -.5\ex@{\kern-\tw@\ex@
				\hbox{\scriptsize #1}\vss}}}}
\makeatother

\DeclareMathOperator*{\argmin}{\arg\!\min}
\DeclareMathOperator*{\argmax}{\arg\!\max}

\newtheorem{theorem}{Theorem}

\newtheorem{proposition}{Proposition}
\newtheorem{corollary}{Corollary}
\newtheorem{remark}{Remark}

\newtheorem{lemma}{Lemma}

\newenvironment{proof}[1][Proof]{\noindent\textbf{#1.} }{\ \rule{0.5em}{0.5em}}

\newif\ifRevComments
\newcommand{\revAdd}[1]{{\ifRevComments\color{teal}\fi#1}}

\begin{document}
\title{Blind Unwrapping of Modulo Reduced Gaussian Vectors: Recovering MSBs from LSBs}

\author{Elad Romanov and Or Ordentlich
	\thanks{}
	\thanks{Elad Romanov and Or Ordentlich are with the School of Computer Science and Engineering, Hebrew University of Jerusalem, Israel (emails: \{elad.romanov,or.ordentlich\}@mail.huji.ac.il). }
	\thanks{This work was supported, in part, by ISF under Grants 1791/17 and 1523/16 and by the GENESIS consortium via the Israel Ministry of Economy and Industry.
	ER acknowledges partial support from the  HUJI Leibniz center and an Einstein-Kaye fellowship. The material in this paper was presented in part at the 2019 International Symposium on Information Theory.}
	\thanks{}}

\date{}

\maketitle

\begin{abstract}
We consider the problem of recovering $n$ i.i.d samples from a zero mean multivariate Gaussian distribution with an unknown covariance matrix, from their modulo wrapped measurements, i.e., measurement where each coordinate is reduced modulo $\Delta$, for some $\Delta>0$. For this setup, which is motivated by quantization and analog-to-digital conversion, we develop a low-complexity iterative decoding algorithm. We show that if a benchmark informed decoder that knows the covariance matrix can recover each sample with small error probability, and $n$ is large enough, the performance of the proposed blind recovery algorithm closely follows that of the informed one. We complement the analysis with numeric results that show that the algorithm performs well even in non-asymptotic conditions.
\end{abstract}

\section{Introduction}
\label{sec:intro}

Let $\bX_1,\ldots,\bX_n\stackrel{i.i.d.}{\sim}\m{N}(\mathbf{0},\bSigma)$ be $n$ i.i.d. realizations of a zero-mean $K$-dimensional Gaussian random vector with covariance matrix $\bSigma\in\RR^{K\times K}$. Let $\bX^*_i$ be the $K$-dimensional vector obtained by reducing each coordinate of $\bX_i$ modulo $\Delta$, for some $\Delta>0$. This paper studies the problem of \emph{blindly} reconstructing the original $n$ samples $\{\bX_1,\ldots,\bX_n\}$ from their wrapped counterparts $\{\bX^*_1,\ldots,\bX^*_n\}$, where the term \emph{blind} refers to \emph{without knowing the covariance matrix $\bSigma$}. See Figure~\ref{fig:scatter} for an illustration. As the modulo operation can be thought of as discarding the most significant bits (MSBs) in the binary representation of each coordinate and keeping only the least significant bits (LSBs), this problem can be alternatively thought of as that of blindly recovering the MSBs of the coordinates from their LSBs.

Since the modulo operation is not invertible, it should be clear that even in the \emph{informed} case, where $\bSigma$ is known, reconstruction of $\bX$ from $\bX^*$ can never be guaranteed to  succeed, and there is always a finite error probability associated with the reconstruction process.\revAdd{\footnote{\revAdd{In the sequel, we measure the reconstruction error with respect to \emph{exact} recovery - precise details are given in the next section. Note that this makes sense, as each measurement $\bX^*$ can correspond to only a discrete set of true signals $\bX$, namely, all the vectors inside the same coset modulo $\Delta$.} }}
 The error probability in the informed case depends on the interplay between the covariance matrix $\bSigma$ and the modulo size $\Delta$. As a simple example, consider the one-dimensional case $X\sim\m{N}(0,\sigma^2)$, where the reconstruction that minimizes the error probability is $\hat{X}(X^*)=X^*$, regardless of $\sigma$, and its error probability is large when $\sigma\gg\Delta$. For general $K$, previous work~\cite{oe13b,de17,othsw18} have demonstrated that in the informed case, the error probability $\epsilon$ in reconstructing $\bX$ from $\bX^*$ typically depends on the ratio $\frac{\Delta}{|\bSigma|^{1/2K}}$, where $|\bSigma|\triangleq \det(\bSigma)$, rather than on the ratio $\frac{\Delta}{\max_{k\in[K]}\sqrt{\bSigma_{kk}}}$. As blind recovery algorithms cannot do better than informed ones, our goal is to develop a blind algorithm that successfully recovers $\{\bX_1,\ldots,\bX_n\}$ from $\{\bX^*_1,\ldots,\bX^*_n\}$ with high probability, whenever $\Delta$ and (the unknown) $\bSigma$ are such that an informed algorithm can achieve a high successful recovery probability.

The main motivation for studying the blind unwrapping problem comes from recent trends in the study of analog-to-digital converters (ADCs). In many emerging applications in communication and signal processing one needs to digitize highly correlated analog processes, where each process is observed by a separate ADC. As a representative, but by no means exclusive, example, consider the front-end of a massive multiple-input multiple-output (MIMO) receiver, where the number of antennas can be of the order of tens and even hundreds, whereas the number of users it serves is moderate, making the signals observed by the various antennas highly correlated~\cite{mcdcm13,hc17}. In theory, if we could use an ADC that \emph{jointly} quantizes all processes, the correlation between them could be exploited in order to reduce the total number of required quantization bits~\cite{berger71,gershogray}. In practice, however, analog-to-digital conversion is implemented by mixed-circuits, which renders joint quantization of the correlated signals impractical. Consequently, it was recently proposed to use the so-called \emph{modulo ADCs} in such scenarios~\cite{oe13b,othsw18} (see also~\cite{bkr17,bkr18,boufounos12,vb16,mjg18}). A modulo ADC is a device that first reduces its input modulo $\Delta$, and only then quantizes it. See Figure~\ref{fig:modADC}. The modulo reduction limits the dynamic range of the resulting signal to the interval $\left[-\frac{\Delta}{2},\frac{\Delta}{2}\right)$, which means that when $\Delta$ is small, one can quantize the wrapped signal to within a good precision using only a few bits. As high correlation between the signals observed by various modulo ADCs allows for correct reconstruction even with small $\Delta$, this architecture successfully exploits the correlation between the signals for reducing the burden from the ADCs. The works~\cite{oe13b,de17,othsw18} studied the performance of modulo ADCs for correlated signals, under various assumptions, and have demonstrated that it is typically quite close to the best performance one could attain by jointly quantizing the signals. However, these works assumed that the decoder that reconstructs the original samples from their wrapped version is informed of $\bSigma$. In practice, one has no access to $\bSigma$ and the blind setup is more appropriate for analog-to-digital conversion.

\begin{figure}[t]
\begin{center}
\includegraphics{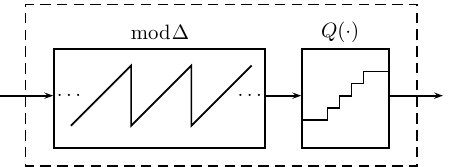}	
\end{center}
\caption{A schematic illustration of the modulo ADC.}
\label{fig:modADC}
\end{figure} 

As another motivation for the blind unwrapping problem, consider the problem of communication in the presence of Gaussian noise with unknown statistics. More precisely, let $\bX_t=\bS_t+\bZ_t$, $t=1,\ldots,n$, where $\bS_t\in\RR^K$ is the channel's input, $\bZ_t\stackrel{i.i.d.}{\sim}\m{N}(\mathbf{0},\bSigma)$  is additive noise statistically independent of $\{\bS_t\}$, and $\bX_t\in\RR^K$ is the channel's output. The goal is to decode $\{\bS_1,\ldots,\bS_n\}$ blindly, i.e., without knowing $\bSigma$. If for all $t$ it holds that $\bS_t\in\Delta\cdot\ZZ^K$, which merely corresponds to using pulse amplitude modulation (PAM), we have that $\bX^*_t=\bZ^*_t$. Thus, solving the blind unwrapping problem correctly, i.e., recovering $\{\bZ_1,\ldots,\bZ_n\}$ from $\{\bZ^*_1,\ldots,\bZ^*_n\}$, corresponds to decoding $\{\bS_1,\ldots,\bS_n\}$, as $\bS_t=\bX_t-\bZ_t$.

The main contribution of this work is an iterative algorithm, with complexity \revAdd{$\m{O}\left(n^2\log{K}+n\mathrm{poly}(K) \right)$ or $\m{O}\left(n^2\log{K}+n\mathrm{poly}(K)+\left(\frac{5}{4} \right)^{K^3/4}\right)$, depending on whether a lattice reduction step within the algorithm is computed exactly of approximated using the LLL algorithm~\cite{lll82},}  for the blind unwrapping problem.\footnote{\revAdd{Much of the literature on communication/estimation schemes based on variants of lattice reduction algorithms completely ignores the computational cost of lattice reduction, as it can always be made negligible, provided that $n$ is sufficiently large. See discussion in~\cite[pages 7662-7663]{zneg12IT}}} Our proposed algorithm builds on the simple observation that for any $\bA\in\ZZ^{K\times K}$ it holds that $[\bA\bX^*]^*=[\bA\bX]^*$ (see~\cite{oe13b} and Section~\ref{subsec:IFD}). Thus, we may use our measurements to compute $\bV^*_i=[\bA\bX^*_i]^*$, $i=1,\ldots,n$, where $\bV_i=\bA\bX_i$. Assuming further that $\bA$ is invertible, we see that if we can blindly recover $\{\bV_1,\ldots,\bV_n\}$ from $\{\bV^*_1,\ldots,\bV^*_n\}$, then, we can solve the original problem of recovering $\{\bX_1,\ldots,\bX_n\}$, simply by setting $\bX_i=\bA^{-1}\bV_i$. In particular, if $\bA$ is such that $\Pr(\bA\bX\notin\CUBE)=\epsilon$, where $\CUBE\triangleq\left[-\frac{\Delta}{2},\frac{\Delta}{2}\right)^K$, the problem of recovering $\bV=\bA\bX$ from $\bV^*$ with success probability $1-\epsilon$ is trivially solved by the estimate $\hat{\bV}=\bV^*$. If $\bSigma$ were known, i.e., in the informed case, the invertible integer matrix that minimizes $\Pr(\bA\bX\notin\CUBE)$ could have been directly computed. Since this is impossible in the blind setup, our algorithm starts with $\bA=\bI$ and iteratively updates its choice of $\bA$, such that $\Pr(\bA\bX\notin\CUBE)$ is reduced from iteration to iteration. It does so by computing at each iteration an estimate $\bSigmaV$ of the truncated covariance matrix $\mathbb{E}[\bA\bX\bX^T\bA^T|\bA\bX\in\CUBE]$, finding the best invertible integer matrix $\tilde{\bA}$ with respect to $\bSigmaV$, and updating $\bA\leftarrow\tilde{\bA}\bA$. 

By establishing and leveraging various properties of truncated Gaussian vectors, we prove that if in all iterations $\bSigmaV$ is an accurate enough estimate for $\mathbb{E}[\bA\bX\bX^T\bA^T|\bA\bX\in\CUBE]$, the algorithm converges to a solution $\hat{\bA}\in\ZZ^{K\times K}$ for which $\Pr(\hat{\bA}\bX\notin\CUBE)$ is close to $\Pr(\bA^{\text{opt}}\bX\notin\CUBE)$, where $\bA^{\text{opt}}$ is the best choice for $\bA$ in the informed case. The speed of convergence is bounded by $\m{O}\left(\frac{1}{\Pr(\bX\in\CUBE)}\right)$. \revAdd{Our analysis assumes the more complicated variant of the algorithm, which exactly solves the lattice reduction problem is used. The benchmark informed decoder is also assumed to exactly solve the lattice reduction problem. In our numerical experiments, on the other hand, both our algorithm and the benchmark informed decoder, use the LLL algorithm~\cite{lll82}.}

Our procedure for estimating $\mathbb{E}[\bV\bV^T|\bV\in\CUBE]$ for $\bV\sim\m{N}(\mathbf{0},\bSigma)$ given i.i.d. wrapped measurements $\{\bV^*_1,\ldots,\bV^*_n\}$ is based on identifying the points in $\{\bV^*_1,\ldots,\bV^*_n\}$ for which $\bV^*_i=\bV_i$, i.e, the points that were not effected by the modulo wrapping, and then computing their empirical covariance matrix. In order to identify the unwrapped points, we observe that if there exists an (informed) estimator for $\bV$ from $\bV^*$ with high success probability, then for two points $\bV^*_i,\bV^*_j$ such that $\bV^*_i=\bV_i$ but $\bV^*_j\neq\bV_j$, it must hold that $\|\bV^*_i-\bV^*_j\|>d$, for some $d$ we explicitly specify. Thus, we construct a graph with nodes $\{\mathbf{0},\bV^*_1,\ldots,\bV^*_n\}$, where an edge between two nodes exists iff the Euclidean distance between them is less than $d$, and identify the unwrapped points with the connected component of $\mathbf{0}$. This set should include only points that were not wrapped, due to our observation. On the other hand, if $n$ is large enough, most of the unwrapped points would belong to the connected component of $\mathbf{0}$. Consequently, we prove that provided that the success probability of the best informed estimator of $\bV$ from $\bV^*$ is high enough, and the number of samples $n$ is high enough, our procedure gives an accurate estimate of  $\mathbb{E}[\bV\bV^T|\bV\in\CUBE]$ for $\bV\sim\m{N}(\mathbf{0},\bSigma)$. Thus, under those conditions our proposed blind algorithm will correctly recover $\{\bX_1,\ldots,\bX_n\}$ from $\{\bX^*_1,\ldots,\bX^*_n\}$ with error probability not much greater than that of the informed decoder.

We complement our analysis with some numerical experiments. The experiments show that our algorithm performs surprisingly well even when $n$ and the informed error probability are quite moderate. In particular, the experiments show that the number of measurements required for successful blind unwrapping is dictated by the number of strong eigenvalues in $\bSigma$, rather than the dimension of $\bSigma$. In particular, for problem of a ``sparse'' nature, such as massive MIMO for example, the algorithm performs well with a small number of measurements.

The structure of the paper is as follows. In Section~\ref{sec:prob} we define the problem, present our algorithm, and state our main analytic results. A quick recap of the informed unwrapping problem is given in Section~\ref{sec:informed}, and the benchmark to which we compare our results is introduced. Numerical results are brought in Section~\ref{sec:sims}. The remainder of the paper is dedicated to the proofs of the main result. 
The analysis of a genie-aided algorithm, where in each iteration we have access to a genie that provides an accurate estimate of the covariance matrix of the truncated Gaussian, is provided in Section~\ref{sec:genieaided}. Then, in Section~\ref{sec:TruncEst} we analyze the performance of the procedure that mimics this genie, i.e., estimates the covariance of a truncated Gaussian vector, given wrapped i.i.d. measurements. The results of the two sections are then combined in Section~\ref{sec:combined}, to yield performance guarantees of the proposed blind recovery algorithm. The paper concludes in Section~\ref{sec:summary}.

\section{Problem Statement, Proposed Algorithm, and Main Results}
\label{sec:prob}

Define the modulo operation
\begin{align}
x^*=[x]\bmod\Delta\triangleq x-\Delta\left\lceil \frac{x}{\Delta}\right\rfloor\in\left[-\frac{\Delta}{2},\frac{\Delta}{2}\right)
\end{align}
where $\lceil t\rfloor\triangleq\argmin_{b\in\ZZ}|t-b|$ is the ``round'' operation, which returns the closest integer to $t$, with the convention that $\lceil a+\frac{1}{2}\rfloor=a$ for $a\in\ZZ$. For a vector $\bx\in\RR^K$, we write $\bx^*$ for the vector obtained by applying the modulo operation on each coordinate of $\bx$, i.e.,
\begin{align}
\bx^*\triangleq[x_1^* \ \cdots \ x_K^*]^T.
\end{align}

Let $\bSigma\in\RR^{K\times K}$ be a positive definite covariance matrix, and let $\{\bX_1,\ldots,\bX_n\}\stackrel{i.i.d.}{\sim} \m{N}(\mathbf{0},\bSigma)$. In the \emph{Blind Unwrapping Problem} we are given only the modulo reduced random vectors $\{\bX^*_1,\ldots,\bX^*_n\}$ and our goal is to recover their unfolded versions $\{\bX_1,\ldots,\bX_n\}$, \emph{without prior knowledge on the covariance matrix $\bSigma$}. In particular, we are interested in devising an algorithm whose input is $\{\bX^*_1,\ldots,\bX^*_n\}$ and whose output is a set of $n$ estimates $\{\hat{\bX}_1,\ldots,\hat{\bX}_n\}$. The performance of an algorithm is measured by
\begin{align}
P_e\triangleq \Pr\left(\{\hat{\bX}_1,\ldots,\hat{\bX}_n\}\neq \{\bX_1,\ldots,\bX_n\}\right).\label{eq:PeDef}
\end{align}
Note that we could have equivalently used the average number of incorrect reconstructions as the performance metric. There is some benefit in considering block-error rate as in~\eqref{eq:PeDef}, as modulo ADCs are often used in conjunction with a prediction filter, and then a single error may have a disastrous effect due to error propagation~\cite{othsw18}. 

\begin{figure*}
\begin{center}
\subfloat[Original]{
			\includegraphics[width=0.45\textwidth]{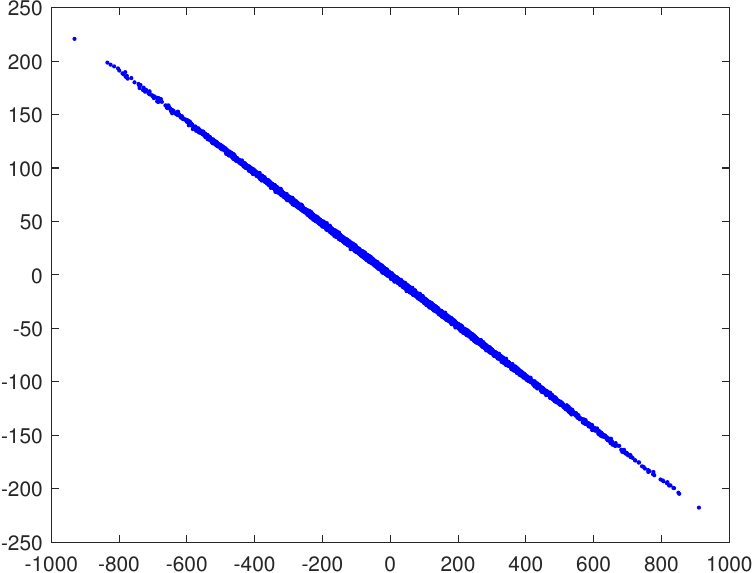}}
		\qquad
\subfloat[Folded]{
			\includegraphics[width=0.45\textwidth]{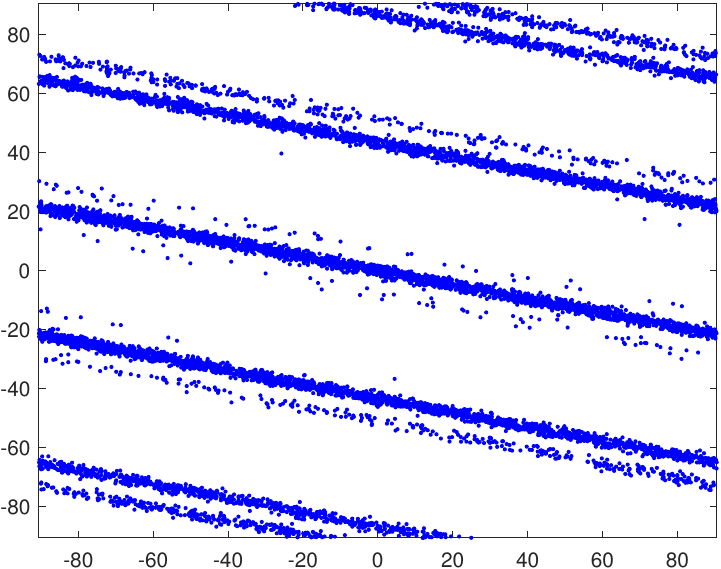}}
\end{center}
\caption{An illustration of the blind unwrapping problem. The points $\{\bX_1,\ldots,\bX_n\}$ are sampled independently from the $\m{N}(\mathbf{0},\bSigma)$ distribution, a typical scatter plot (for a spiky covariance matrix $\bSigma$) is shown in left hand side of the figure. Given the modulo folded points $\{\bX^*_1,\ldots,\bX^*_n\}$, whose scatter plot is given in right hand side of the figure, and no prior knowledge of $\bSigma$, our goal is to recover $\{\bX_1,\ldots,\bX_n\}$.}
\label{fig:scatter}
\end{figure*}

For a $K$-dimensional random vector $\bX~\sim\m{N}(\mathbf{0},\bSigma)$ and a set $\m{S}\subset\RR^K$ \revAdd{of positive Lebesgue measure}, we define the truncated random vector $\bY\sim[\bX|\bX\in\m{S}]$, whose probability density function (pdf) is
\begin{align}
f_{\bY}(\by)=\frac{f_{\bX}(\by)\Ind_{\{\by\in\m{S}\}}}{\Pr(\bX\in\m{S})}=\frac{e^{-\frac{1}{2}\by^T\bSigma^{-1/2}\by}}{\int_{\bt\in\m{S}}e^{-\frac{1}{2}\bt^T\bSigma^{-1/2}\bt}d\bt}\Ind_{\{\by\in\m{S}\}}.
\end{align}
\revAdd{In other words, the law of $\bY$ is the law of $\bX$ conditioned on $\bX\in \m{S}$.}
We denote the covariance matrix of the truncated random vector by
\begin{align}
\mathbb{E}[\bY\bY^T]=\mathbb{E}[\bX\bX^T|\bX\in\m{S}].
\end{align}
We propose a low-complexity algorithm for the blind unfolding problem. The algorithm is based on the following simple observation, which holds due to the fact \revAdd{that} the modulo operation is invariant with respect to translation by integer copies of $\Delta$. 

\begin{proposition}
	For any integer vector $\ba\in\ZZ^K$ and and $\bx\in\RR^K$ it holds that
	\begin{align}
	\revAdd{ [\ba^T\bx^*]^*=[\ba^T\bx]^*. }
	\end{align}
	\label{prop:integermod}
\end{proposition}


The algorithm iterates between two basic procedures: 1)Estimating the covariance matrix $\overset{\vee}{\bSigma}$ of the Gaussian random vector $\bX$ truncated to the set $\m{S}=\CUBE$, from the measurements $\{\bX^*_1,\ldots,\bX^*_n\}$. 2)Finding a full-rank integer matrix $\bA$, based on $\overset{\vee}{\bSigma}$, such that $\Pr(\bA\bX\notin\CUBE)<\Pr(\bX\notin\CUBE)$, provided that $\Pr(\bX\notin\CUBE)$ was not very small to begin with, and updating the measurements set to $\{[\bA\bX^*_1]^*,\ldots,[\bA\bX^*_n]^*\}=\{[\bA\bX_1]^*,\ldots,[\bA\bX_n]^*\}$. In order to facilitate the analysis of the algorithm in the sequel, we will restrict attention to a subset of the full-rank integer matrices, namely the group of \emph{unimodular} matrices $\SLk$, consisting of all matrices in $\ZZ^{K\times K}$ with determinant $1$ or $-1$. Note that $\bA\in\SLk$ implies that $\bA^{-1}\in\SLk$.

Below we give the precise algorithm. The algorithm has two parameters: $d\in\RR_+$ and $M\in\mathbb{N}$, that are chosen by the designer, according to considerations discussed in Sections~\ref{sec:genieaided} and~\ref{sec:TruncEst}. The main algorithm makes use of the procedure $\mathrm{EstimateTruncatedCovariance}$ which will be described immediately. 

\textbf{\textit{Inputs}}: $(\bX^*_1,\ldots,\bX^*_n)$, $\Delta$, and two design parameters $d$ and $M$.

\textbf{\textit{Main Algorithm}}:
\begin{itemize}
	\item Initialization: $\bA=\bI$, $\tilde{\bA}=\mathbf{0}^{K\times K}$, $\mathrm{ctr}=0$, $(\bV^*_1,\ldots,\bV^*_n)=(\bX^*_1,\ldots,\bX^*_n)$
	\item While $\tilde{\bA}\neq\bI$ and $\mathrm{ctr}<M$
	\begin{enumerate}
		\item $\overset{\vee}{\bSigma}=\mathrm{EstimateTruncatedCovariance}\left((\bV^*_1,\ldots,\bV^*_n),d\right)$
		\item Compute
		\begin{align}
		\tilde{\bA}&=[\tilde{\ba}_1|\cdots|\tilde{\ba}_K]^T=\argmin_{\bar{\bA}\in\SLk}\max_{k\in[K]} \bar{\ba}_k^T\overset{\vee}{\bSigma}\bar{\ba}_k,\label{eq:step2}
		\end{align} 
		\item Set $\bV^*_j\leftarrow [\tilde{\bA}\bV_j^*]^*=[\tilde{\bA}\bV_j]^*$ for $j=1,\ldots,n$, $\bA\leftarrow \tilde{\bA}\cdot\bA$, and $\mathrm{ctr}\leftarrow \mathrm{ctr}+1$
	\end{enumerate}
\end{itemize}

\textbf{\textit{Outputs}}: $\bA$ and the estimates $\hat{\bX}_j=\bA^{-1}\bV^*_j$, for $j=1,\ldots,n$.

The algorithm $\mathrm{EstimateTruncatedCovariance}\left((\bV^*_1,\ldots,\bV^*_n),d\right)$, which is used within the main algorithm, is as follows.

\textbf{\textit{Inputs}}: $(\bV^*_1,\ldots,\bV^*_n)$,  and a design parameter $d$.

\textbf{\textit{EstimateTruncatedCovariance Algorithm}}:
\begin{enumerate}
	\item Set $\bV^*_0=\mathbf{0}$
	\item Construct a graph where each of the $n+1$ points $(\bV^*_0,\bV^*_1,\ldots,\bV^*_n)$ is a vertex, and an edge between $\bV^*_i$ and $\bV^*_j$ exists iff $\|\bV^*_i-\bV^*_j\|_2<d$
	\item Find the connected component of $\bV^*_0=\mathbf{0}$, and denote it by $\m{T}$
	\item If $|\m{T}|< K+1$, set $d\leftarrow 1.1d$ and return to step $2$; else
\end{enumerate}
 
\textbf{\textit{Output}}: set
\begin{align}
\overset{\vee}{\bSigma}=\frac{1}{|\m{T}|}\sum_{\bt\in\m{T}}\bt\bt^T,
\end{align}
as the estimate of $\mathbb{E}\left[\bV\bV^T|\bV\in\CUBE \right]$.

Figure~\ref{fig:example} illustrates the progress of the algorithm from iteration to iteration, for $K=2$. Note that since all points in $\CUBE$ are at distance at most $\sqrt{K}\Delta/2$ from the origin, the procedure $\mathrm{EstimateTruncatedCovariance}$ must converge after at most $\log\left(\sqrt{K}\Delta/2d\right)/\log(1.1)$ iterations. The computational complexity of the procedure $\mathrm{EstimateTruncatedCovariance}$ is therefore $\m{O}(n^2\log{K}+n\mathrm{poly}(K))$. Step $2$ of the main algorithm corresponds to solving the shortest basis problem (SBP) of a lattice, \revAdd{specifically, of the lattice $\Lambda(\bSigma^{1/2})=\bSigma^{1/2}\cdot \ZZ^K$. Recall that that a basis of $\Lambda(\bSigma^{1/2})$ is a set of $K$ linearly independent lattice vectors whose span (with respect to $\ZZ$) is the entire lattice; in other words, a basis consists of the columns of $\bSigma^{1/2}\bA^T$, where $\bA\in \SLk$ is any unimodular matrix. The SBP asks for a basis whose largest vector has the least possible norm. This is precisely the problem solved in~\eqref{eq:step2}.} The computational complexity of an exact solution of SBP is at least exponential in $K$~\cite{bs99}, and is known to be at most $\m{O}\left(\left(\frac{5}{4}\right)^{K^3/4}\right)$~\cite{helfrich1985}. However, one can always approximate the solution of SBP using efficient sub-optimal algorithms, as the LLL algorithm~\cite{lll82}, reducing the computational complexity of step $2$ of our algorithm to $\mathrm{poly}(K)$. Finally, the complexity of step $3$ of the main algorithm is $n\mathrm{poly}(K)$. Thus, for fixed $M$ and $d$ whose value does not decrease with $K$, we have that the computational complexity of the algorithm is $\m{O}\left(n^2\log{K}+n\mathrm{poly}(K)+\left(\frac{5}{4}\right)^{K^3/4}\right)$ if the SBP~\eqref{eq:step2} is computed exactly, or $\m{O}\left(n^2\log{K}+n\mathrm{poly}(K)\right)$ when it is approximated using the LLL algorithm. In practice, the latter choice is more attractive, and all the simulations performed in this paper indeed used the LLL algorithm to solve~\eqref{eq:step2}.
\revAdd{We stress, however, that the analysis of the proposed algorithm, in the sequel, relies on the availability of an \emph{exact} solution the the SBP in step 2. The benchmark informed decoder, to which the performance of the proposed algorithm is compared, is also assume to compute an exact solution of the SBP.}


\begin{figure*}
	\begin{center}
		\subfloat[Iteration 1]{
			\includegraphics[width=0.45\textwidth]{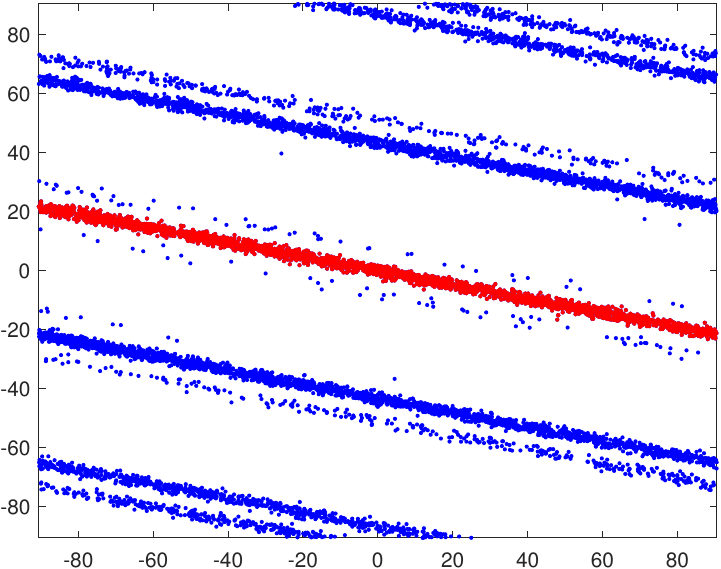}}
		\qquad
		\subfloat[Iteration 2]{
			\includegraphics[width=0.45\textwidth]{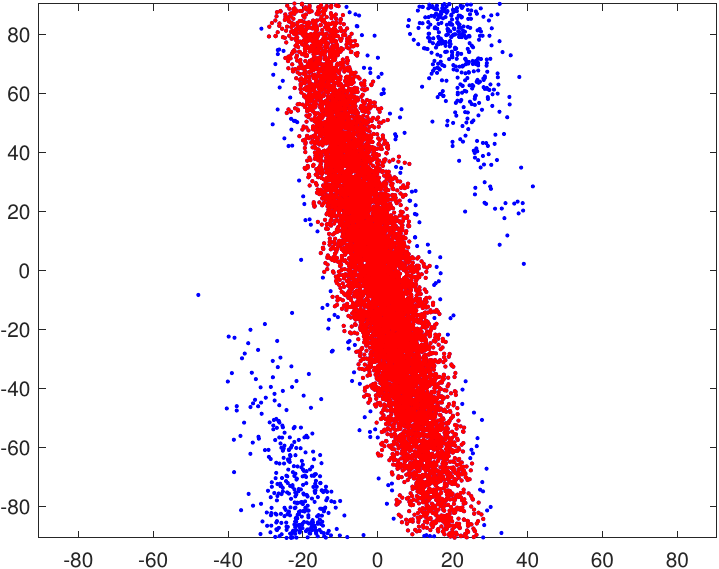}}
		\qquad
		\subfloat[Iteration 3]{
			\includegraphics[width=0.45\textwidth]{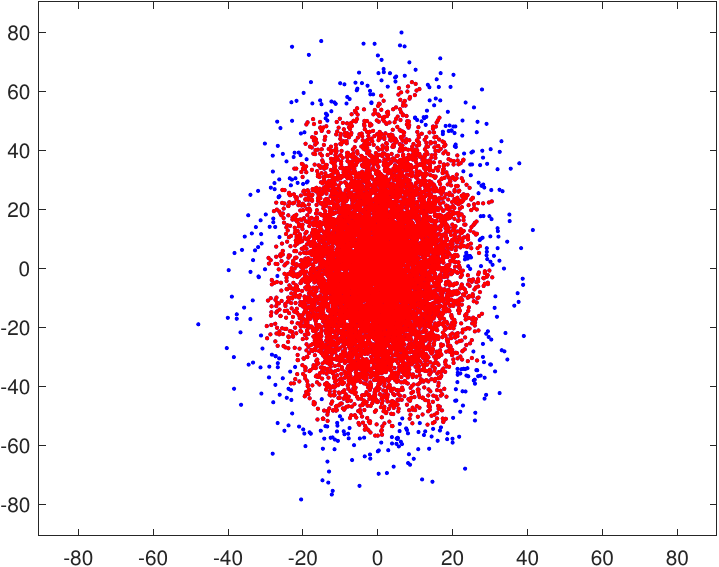}}
		\qquad
		\subfloat[Output]{
			\includegraphics[width=0.45\textwidth]{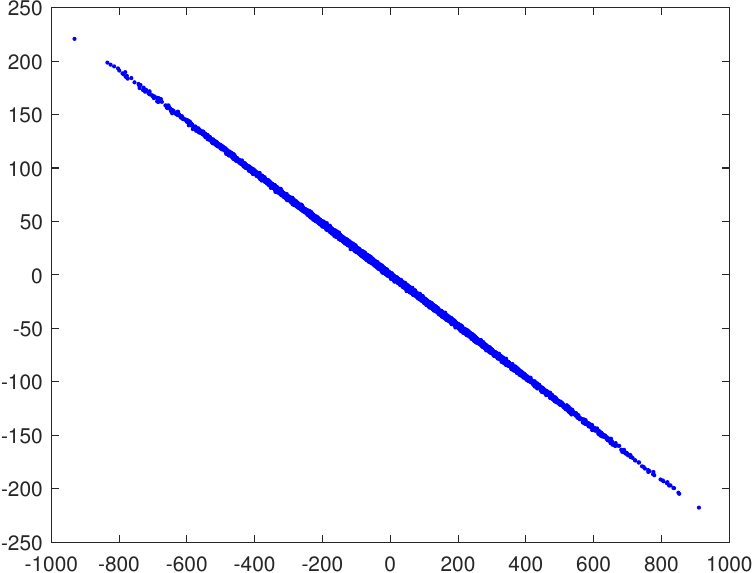}}
	\end{center}
\label{fig:example}
\caption{An illustration of the progress of \textbf{main algorithm} from iteration to iteration. For the $i$th iteration, we show the scatter plot of $\left([\bA^{(m)}\bV_1]^*,\ldots,[\bA^{(m)}\bV_n]^*\right)$, where $\bA^{(m)}$ is the estimated integer matrix $\bA$ before the update made at the end of the $m$th iteration. The red points correspond to the points belonging to the connected component of $\mathbf{0}$ in the graph constructed by the $\mathrm{EstimateTruncatedCovariance}$ procedure. The bottom right figure corresponds to the scatter plot of the final estimates $\{\hat{\bX}_1,\ldots,\hat{\bX}_n\}$ the algorithm outputs.}
\end{figure*}

The next sections will be mostly devoted to performance analysis of the algorithm above. The analysis requires several assumptions on the underlying covariance matrix, which we now specify. To state those assumptions, we first need a few definitions. For a symmetric positive definite matrix $\bSigma$, denote the eigenvalues by $\lambda_1(\bSigma)\geq\cdots\geq\lambda_K(\bSigma)>0$, and consider the eigendecomposition $\bSigma=\bU\bD\bU^T$ of $\bSigma$, where $\bU\in\RR^{K\times K}$ is a unitary matrix, and $\bD=\diag(\lambda_1(\bSigma),\ldots,\lambda_K(\bSigma))$. Let $\bSigma^{-1/2}=\bU\bD^{-1/2}\bU^T$, where $\bD^{-1/2}=\diag\left(\frac{1}{\sqrt{\lambda_1(\bSigma)}},\ldots,\frac{1}{\sqrt{\lambda_1(\bSigma)}}\right)$. For a full-rank matrix $\bG\in\RR^{K\times K}$ we define the lattice $\Lambda(\bG)=\bG\cdot\ZZ^K$. The packing radius of a lattice $\Lambda(\bG)$ is defined as $r_0(\Lambda(\bG))=\frac{1}{2}\min_{\bb\in\ZZ^{K}\setminus\{\mathbf{0}\}}\|\bG\bb\|$ and the effective radius $r_{\text{eff}}(\Lambda(\bG))$ is defined as the radius of a $K$-dimensional Euclidean ball whose volume is $|\bG|$, i.e., $V_K r^K_{\text{eff}}(\Lambda(\bG))=|\bG|$, where $V_K$ is the volume of a $K$-dimensional unit ball. 

Let $\bSigma\in\RR^{K\times K}$ be positive-definite, and let $\bX\sim\m{N}(\mathbf{0},\bSigma)$. Let $0<\epsilon<1$, $\tau_{\text{min}}>0$, $0<P<1$ and $0<\rho_{\text{pack}}<1$ be some given numbers.
We say that $\bSigma$ satisfies assumption $Ai$, $i=1,\ldots,4$,  if
\begin{alignat}{2}
A1 \ &: \ \min\left\{\Pr(\bA\bX\notin\CUBE) \ : \ \bA\in\SLk \right\}\leq \epsilon\\
A2 \ &: \ \lambda_K(\bSigma)\geq\tau_{\text{min}}^2\\
A3 \ &: \ \Pr(\bX\notin\CUBE)\leq P \\
A4 \ &: \ \frac{r_0(\Lambda(\bSigma^{1/2}))}{r_{\text{eff}}(\Lambda(\bSigma^{1/2}))}\geq \rho_{\text{pack}}.
\end{alignat}

In Section~\ref{sec:informed} we will review the informed integer-forcing decoder, and it will become clear that assumption $A1$ simply means that the informed integer-forcing decoder can recover $\bX$ from $\bX^*$ with error probability at most $\epsilon$. Assumption $A1$ may at first glance seem somewhat ad-hoc. Indeed, a more reasonable assumption would be
\begin{align}
\tilde{A}1 \ &: \ \exists g:\left[-\frac{\Delta}{2},\frac{\Delta}{2}\right)^K\to \RR^K \ : \ \Pr(g(\bX^*)\neq\bX)\leq \epsilon.
\end{align}
However, as demonstrated in~\cite{oe13b,de17,othsw18} (see also Remark~\ref{rem:ifloss}), $A1$ is typically a good proxy for $\tilde{A}1$, and will be a convenient choice for the analysis in the sequel. Assumption $A2$ corresponds to the random vector $\bX$ having non-negligible energy in all directions. This assumption is well justified in the context of modulo-ADCs \revAdd{(with substractive dithers)}, as the wrapped Gaussian vector that needs to be recovered there is already quantized. Thus, thinking of the quantizer as an additive white noise source with variance $D$, we see that assumption $A2$ must hold with $\tau^2_{\text{min}}=D$; \revAdd{see also \cite{othsw18}}. Assumption $A3$ requires that the probability of $\bX$ missing $\CUBE$ is bounded \revAdd{above by some number $P<1$}. Indeed, if the probability of missing $\CUBE$ is arbitrarily close to $1$, it is difficult to estimate $\mathbb{E}[\bX\bX^T|\bX\in\CUBE]$ to a good precision from a finite number of i.i.d. samples distributed as $\bX^*$. Assumption $A4$ is needed in order to control (upper bound) the ratio $\Delta/|\bSigma|^{1/2K}$. If $A1$ holds with small $\epsilon$, in the context of modulo-ADCs, assumption $A4$ essentially means that the quantization rate is not much greater than the source's rate-distortion function. 

There is some redundancy in assumptions $A1-A4$. It can be shown that if $\bSigma$ satisfies $A1$ and $A2$, it immediately implies that it satisfies $A3$ with some $P(\epsilon,\tau_{\text{min}},\Delta)$. This follows since, as we show in Proposition~\ref{prop:detbound}, the determinant $|\bSigma|$ can be upper bounded in terms of $\epsilon$, and combining with assumption $A2$, this gives an upper bound on the maximal eigenvalue of $\bSigma$, which immediately lends itself to an upper bound on $\Pr(\bX\notin\CUBE)$. On the other hand, in proposition~\ref{prop:chi1_bound} we also show that assumptions $A1$, $A3$ and $A4$ imply that assumption $A2$ holds with some $\tau_{\text{min}}(\epsilon,P,\rho_{\text{pack}})$. Despite this redundancy, we have chosen to present our results in terms of all four assumptions, as this leads to clearer expressions. 

\revAdd{Let 
\begin{equation}
	Q(t)\triangleq\frac{1}{\sqrt{2\pi}}\int_{t}^\infty e^{-\frac{t^2}{2}}dt
\end{equation}
be the Q-function, and $Q^{-1}(t)$ be its inverse. Define the function
\begin{equation}\label{eq:f}
	f(\alpha)\triangleq\alpha^2-1-\ln{\alpha^2} \,.
\end{equation}}
Our main analytic results are the following.
\begin{theorem}[Guarantee on Genie-Aided Algorithm]
Let $\bX\sim\m{N}(\mathbf{0},\bSigma)$, where $\bSigma$ satisfies assumptions $A1$ and $A3$, and $\epsilon$ satisfies $Q^{-1}\left(\frac{\epsilon}{2}\right)\geq 6\sqrt{K}$. Consider a genie-aided version of the main algorithm, where the procedure\\ \mbox{$\mathrm{EstimateTruncatedCovariance}((\bV^*_1,\ldots,\bV^*_n),d)$} is replaced with a genie that returns a matrix $\bSigmaV$ that satisfies
\begin{equation}
\begin{split}
	(1-\beta)&\mathbb{E}[\bV\bV^T|\bV\in\CUBE]\\
	&\preceq\bSigmaV 
	\preceq (1+\beta)\mathbb{E}[\bV\bV^T|\bV\in\CUBE],
\end{split}
\label{eq:estassumption}
\end{equation}
for some $0\leq \beta\leq 0.1$. Then, after 
\begin{align}
M=\frac{\log{180}}{\log\left(\frac{4}{3+P}\right)}+2
\end{align}
iterations, the matrix $\bA$ found by the algorithm must satisfy
\begin{align}
\Pr(\bA\bX\notin\CUBE)\leq K\cdot Q\left(0.99\sqrt{\frac{1-\beta}{1+\beta}}\cdot Q^{-1}\left(\frac{\epsilon}{2}\right) \right).
\end{align}
\label{thm:genie}
\end{theorem}

\begin{theorem}[Guarantee on truncated covariance estimation]
	\label{thm:cov_est_ugly_bnd}
	Let $\bSigma\in\RR^{K\times K}$ satisfy assumptions $A1$, $A2$ and $A3$, and $\epsilon$ is such that $\kappa_{\epsilon}\triangleq \frac{Q^{-1}\left(\frac{\epsilon}{2}\right)}{\sqrt{K}}>1$.
	Suppose that the procedure $\mathrm{EstimateTruncatedCovariance}\left((\bX^*_1,\ldots,\bX^*_n),d\right)$ is run with a distance parameter $d=2\eta\sqrt{K}\cdot\tau_{\text{min}}\cdot\kappa_{\epsilon}$, for some $\eta\in \left( 0, 1-\frac{1}{\kappa_{\epsilon}}\right)$,
	and $\bX_1^*,\ldots,\bX_n^*$ are $n$ independent wrapped samples from $\m{N}(\mathbf{0},\bSigma)$. 
	Fix a precision parameter $\beta > 0$, and denote
	\begin{equation}
	\overline{\beta} \triangleq \beta +  \frac{2K \cdot \left(\frac{\Delta}{\tau_{\text{min}}}\right)^2 }{\alpha(P;K)\cdot \sqrt{n}} \,,
	\end{equation}
	where $\alpha(P;K)\triangleq \max \left[ F_{\chi^2(K+2)}\left( \left( Q^{-1}\left(\frac{P}{2}\right)\right)^2 \right), 1-\sqrt{3P}\right]$,
and $F_{\chi^2(K+2)}$ is the CDF of a (standard) $\chi^2$ random variable with $K+2$ degrees of freedom.
	Then, there is an absolute constant $C>0$ such that if 
	\begin{equation}
	n \ge C\cdot \frac{K\cdot \log K}{2 \cdot \alpha(P;K)} \cdot \left(\frac{\Delta}{\tau_{\text{min}}}\right)^2 \cdot \frac{1}{\beta^2} \,,
	\end{equation}
	then with probability at least $1-p_{\text{est-err}}$, 
	if holds that
	{ \small
	\begin{equation}
		\begin{split}
	\left(1-\overline{\beta}\right)&\mathbb{E}\left[\bX\bX^T|\bX\in\CUBE\right] \\
	&\preceq \bSigmaV 
	\preceq \left(1+\overline{\beta}\right)\mathbb{E}\left[\bX\bX^T|\bX\in\CUBE\right] \,.			
		\end{split}
	\end{equation}
}
	Here
	\begin{equation}
\begin{split}
p_{\text{est-err}} &= \Pr(\m{E}_{\text{false-positive}})+ \Pr \left( \m{E}_{\text{sample-est}} \right) \\
&+ \Pr(\m{E}_{\text{many-escapees}})  + \Pr(\m{E}_{\text{miss-cover}}) \,,
\end{split}
	\end{equation}
	and
	\begin{alignat}{2}
		&\Pr\left(\m{E}_{\text{false-positive}}\right)  \le ne^{-\frac{K}{2}f((1-\eta)\kappa_{\epsilon})}\,,\\
		&\Pr \left(\m{E}_{\text{many-escapees}}  \right) \le e^{-\frac{3}{14}\sqrt{n}}  \,,
	\end{alignat}
{\small	
	\begin{equation}
		\begin{split}
\Pr \left( \m{E}_{\text{sample-est}} \right) 
&\le 
\exp\left[-\beta^2 \cdot \frac{\alpha(P,K)}{C\cdot \frac{K}{2} \cdot \left(\frac{\Delta}{\tau_{\text{min}}}\right)^2}\cdot n \right]
+ e^{-\frac{1}{2}(1-P)^2n} \,,
		\end{split}
	\end{equation}
}
{\small
	\begin{equation}
		\begin{split}
			\Pr \left(\m{E}_{\text{miss-cover}}\right) 
			\le 
			\exp \Biggl[ &-\left(\frac{ \frac{\tau_{\text{min}}}{\abs{\bSigma}^{1/2K}} \cdot \kappa_{\epsilon}\cdot \eta}{\sqrt{\frac{\pi}{2}} \left(\frac{2\Delta}{\tau_{\text{min}}} + 4\right)}\right)^K e^{\frac{1}{2}\log(n) - \sqrt{\frac{1}{2}K\log(n)} - \frac{1}{2}K} \\
			&+K\log\left(\sqrt{K}+\frac{\frac{2\Delta}{\tau_{\text{min}}}}{\kappa_{\epsilon}\cdot \eta}\right)  \Biggr] \,.
		\end{split}
	\end{equation}
}
\end{theorem}

Theorem~\ref{thm:genie} and Theorem~\ref{thm:cov_est_ugly_bnd} will be proved in Sections~\ref{sec:genieaided} and~\ref{sec:TruncEst}, respectively. In Section~\ref{sec:combined} we prove that under certain assumptions the combination of the two theorems above imply that the performance of a slight modification of the proposed algorithm achieves error probability close to that of the informed decoder.

Let us first introduce the modified algorithm. The twist is as follows: given the design parameter $M$, we partition $[n]$ into $M$ disjoint index sets $\m{I}_1,\ldots,\m{I}_M$, each of size $n/M$. In the main algorithm, in step $2$ of the $m$th iteration, instead of applying $\mathrm{EstimateTruncatedCovariance}((\bV^*_1,\ldots,\bV^*_n),d)$ we apply\\ $\mathrm{EstimateTruncatedCovariance}((\bV^*_{i_{m,1}},\ldots,\bV^*_{i_{m,n/M}}),d)$, where \revAdd{$\m{I}_m=\{i_{m,1},\ldots,i_{m,n/M}\}$}. Thus, every sample only participates once (or never participates) in $\mathrm{EstimateTruncatedCovariance}$. The reason for this variation is to avoid intricate statistical dependencies that may develop between the samples as the algorithm progresses, which complicates the analysis. We show that the matrix $\bA$ returned by the modified algorithm is almost as good as $\argmin_{\bA\in\SLk}\Pr(\bA\bX\notin\CUBE)$. In order to control (lower bound) the smallest eigenvalue of $\bA\bSigma\bA^T$ throughout all iterations of the main algorithm, such that Theorem~\ref{thm:cov_est_ugly_bnd} can be applied, we require a refinement of assumption $A1$. In particular we say that $\bSigma$ satisfies assumption $A1^*$ if $\min\left\{\Pr(\bA\bX\notin\CUBE) \ : \ \bA\in\SLk \right\}= \epsilon$. \revAdd{That is, the error probability of the informed decoder is \emph{exactly} $\epsilon$ (rather than upper bounded by it). Clearly, assumption $A1^*$ implies assumption $A1$.}

\begin{theorem}[Success probability for the (modified) main algorithm: finding a good $\bA$]
	\label{thm:success_estimator}
	Let $\bX~\sim\m{N}(\mathbf{0},\bSigma)$, where $\bSigma$ satisfies assumptions $A1^*$, $A3$ and $A4$.
	Suppose that the modified main algorithm is run with parameters
	\begin{equation}
	M\geq\frac{\log{180}}{\log\left(\frac{4}{3+P}\right)}+2
	\end{equation} 
	and
	\begin{equation}	
	d = 0.01 \cdot 2\sqrt{K} \cdot \frac{\Delta}{\chi_1(\epsilon;P,K)} \,,
	\end{equation}
	where 
	\begin{equation}
	\chi_1(\epsilon;P,K) \triangleq \frac{K^{2K-\frac{1}{2}}\cdot V_K \cdot \left(Q^{-1}\left(\frac{\epsilon}{2K}\right)\right)^K  }{2^K\cdot\rho_{\text{pack}}^K\cdot \left(Q^{-1}\left(\frac{P}{2}\right)\right)^{K-1} } \,.
	\end{equation}
	
	\revAdd{Suppose that $n$ is sufficiently large, namely,\footnote{{By this notation, we mean that $c_1 \epsilon^{-\zeta} \le n \le c_2\epsilon^{-\zeta}$ for some fixed numbers $c_1,c_2>0$, as $\epsilon\to 0$.}} {$n\asymp \epsilon^{-\zeta}$} for any $\zeta<(0.99)^2$, and that $\epsilon$ is sufficiently small, in the sense that $\epsilon<\bar{\epsilon}$, where
	\[
	\bar{\epsilon} = \bar{\epsilon}(K,P,\rho_{\text{pack}}, \zeta) > 0\,,
	\]
	is some small threshold. Then with probability at least $1-2 ne^{-\frac{K}{2}f\left(0.99\cdot \frac{Q^{-1}\left(\frac{\epsilon}{2}\right)}{\sqrt{K}}\right)}$, 
	the matrix $\bA$ returned by the algorithm satisfies
	\begin{equation}
	\Pr(\bA\bX\notin\CUBE)\leq K\cdot Q\left(0.98\cdot Q^{-1}\left(\frac{\epsilon}{2}\right) \right),
	\end{equation}
	where the function $f(\cdot)$ is as defined in~\eqref{eq:f}.}
\end{theorem}

The matrix $\bA$ found by the modified algorithm depends on the measurements $\{\bX^*_1,\ldots,\bX^*_n\}$. Theorem~\ref{thm:success_estimator} shows that for a ``fresh'' sample $\Pr(\bA^{-1}[\bA\bX^*]^*\neq\bX)$ is small. However, due to the dependencies between $\bA$ and the sample points, it is not immediate that $\Pr(\bA^{-1}[\bA\bX_i^*]^*\neq\bX_i)$, $i=1,\ldots,n$, is also small. Yet, we prove the following.

\begin{theorem}[Success probability for the (modified) main algorithm: recovering $\bX_1,\ldots,\bX_n$]
	\label{thm:success_recovery}
	Assume the setup of Theorem~\ref{thm:success_estimator}, and let $\bA=\bA(\bX^*_1,\ldots,\bX^*_n)$ be the final matrix returned by the algorithm. Assuming that $\epsilon$ is sufficiently small, we have 
	\begin{equation}\label{eq:thm4}
		\begin{split}
			P_e&\triangleq \Pr\left(\{\hat{\bX}_1,\ldots,\hat{\bX}_n\}\neq \{\bX_1,\ldots,\bX_n\}\right) \\
			&\le 2 ne^{-\frac{K}{2}f\left(0.99\cdot\frac{Q^{-1}\left(\frac{\epsilon}{2}\right)}{\sqrt{K}}\right)} + n e^{-\frac{K}{2}f \left(\frac{Q^{-1}(\epsilon')}{\sqrt{K}}\right)} \,,
		\end{split}
	\end{equation}
	where 
	\begin{equation}
	\epsilon' = K\cdot Q\left(0.98\cdot Q^{-1}\left(\frac{\epsilon}{2}\right) \right) \,,
	\end{equation}
	\revAdd{and $f(\cdot)$ is given in Eq. (\ref{eq:f}).}

\end{theorem}
\revAdd{
	To get a better sense of how the bound in Eq. (\ref{eq:thm4}) behaves when $\epsilon$ is small, use the rough approximations
	\[
	f(\alpha) \approx \alpha^2\,,\quad Q(\alpha) \approx e^{-\alpha^2/2}\,,
	\]
	valid for large $\alpha$, and $Q^{-1}(\epsilon)\approx \sqrt{2\log \frac1\epsilon}$, valid for small $\epsilon>0$, which yields 
	\[
	e^{-\frac{K}{2}f\left((1-\delta)\frac{Q^{-1}(\epsilon)}{\sqrt{K}}\right)} \approx \epsilon^{(1-\delta)^2} \,,
	\]
	for small $\epsilon$. This means that the right-hand-side of Eq. (\ref{eq:thm4}) behaves like 
	\[
	\approx 2n\epsilon^{0.99^2} + \epsilon' \approx 2n\epsilon^{0.99^2}  + Kn\epsilon^{0.98^2} =  \m{O}\left( n\epsilon^{0.96} \right) 
	\]
	as $\epsilon\to 0$.
}
\revAdd{Thus, under the asymptotic assumptions of $\epsilon\to 0$ and $n$ that grows sufficiently fast with $\epsilon$, for example $n\asymp \left(\frac1\epsilon\right)^\zeta$ (see statement of Theorem \ref{thm:success_estimator}), we find that the performance of our blind algorithm is \emph{essentially} as good as that of the informed algorithm. }
Our numerical experiments in Section~\ref{sec:sims} indicate that in various scenarios of practical interest our analysis is over pessimistic, and the algorithm performs as well as the informed benchmark integer-forcing decoder even for moderate values of $\epsilon$ and $n$. For such values, our analysis does not yield useful bounds, mostly due to the large factors involved in $p_{\text{est-err}}$, characterized in Theorem~\ref{thm:cov_est_ugly_bnd}. \revAdd{Specifically, our analysis of the truncated covariance estimation algorithm relies on a covering argument, which we suspect gives somewhat loose bounds.}  In light of this, we believe our analysis should be viewed as an explanation to why the algorithm works, rather than a prediction to the performance it attains.

\section{Recap on Informed Decoding}
\label{sec:informed}

In this section we consider the case of \emph{informed unwrapping}, where the decoder knows the covariance matrix $\bSigma$. First, we show that the optimal decoder, in terms of minimizing the error probability, corresponds to finding the nearest point in a lattice induced by $\bSigma$. We derive an exact expression for the error probability it achieves. Then, we review the so-called \emph{Integer-forcing decoder}, and recall some of its properties, that will be needed in the analysis that follows.

Let $\bX\sim\m{N}(\mathbf{0},\bSigma)$, such that $f_{\bX}(\bx)=\frac{1}{(2\pi)^{K/2}|\bSigma|^{1/2}}e^{-\frac{1}{2}\bx^T\bSigma^{-1}\bx}$ is its pdf, and let $\bSigma^{-1/2}$ be as defined in Section~\ref{sec:prob}. Denote the \revAdd{maximum a-posteriori (MAP)} estimator for $\bX$ from $\bX^*$, i.e., the estimator $g:\CUBE\to\RR^K$, that minimizes $\Pr(g(\bX^*)\neq\bX)$, by
\begin{align}
g_{\text{MAP}}(\bx^*)&\triangleq\bx^*+\Delta\bb_{\text{MAP}}(\bx^*),
\end{align}
where
\begin{align}
\bb_{\text{MAP}}(\bx^*)&\triangleq\argmax_{\bb\in\ZZ^K} f_\bX(\bx^*+\Delta\bb)\nonumber\\
&=\argmin_{\bb\in\ZZ^K} (\bx^*+\Delta\bb)^T\bSigma^{-1}(\bx^*+\Delta\bb)\nonumber\\
&=\argmin_{\bb\in\ZZ^K}\|\bSigma^{-1/2}\bx^*+\Delta\bSigma^{-1/2}\bb\|.\label{eq:bMAP}
\end{align}
\revAdd{Ties in the argmin above are broken in an arbitrary, but systematic, manner. }
Let $\m{R}_{\text{MAP}}\subset\RR^K$ be the set of all points for which the MAP estimator is correct, \revAdd{excluding the boundary of the decision region (which has measure zero)}. 
\revAdd{
Note that $g_{\text{MAP}}(\bx^*)=\bx$ (again, excluding the boundary) if and only if for every other member $\by$ of the same coset (meaning $\bx-\by\in \Delta \ZZ^K$), one has $\bx^T \bSigma^{-1}\bx < \by^T \bSigma^{-1}\by$. In other words,
{\small
\begin{alignat}{2}
\m{R}_{\text{MAP}}
&=\Bigl\{\bx\in\RR^K \ : \  &&\bx^T\bSigma^{-1}\bx<(\bx-\Delta\bb)^T\bSigma^{-1}(\bx-\Delta\bb), \nonumber \\
& &&\ \ \forall\bb\in\ZZ^K\setminus\{\mathbf{0}\} \Bigr\}\nonumber\\
&=\Bigl\{\bx\in\RR^K \ : \ &&\|\bSigma^{-1/2}\bx\|^2<\|\bSigma^{-1/2}\bx-\bSigma^{-1/2}\Delta\bb\|^2, \nonumber \\
& &&\ \ \forall\bb\in\ZZ^K\setminus\{\mathbf{0}\} \Bigr\}\label{eq:preVor}.
\end{alignat}}
}
Let 
\begin{align}
\Lambda=\Lambda(\Delta\bSigma^{-1/2})=\Delta\bSigma^{-1/2}\ZZ^K
\end{align}
be the lattice with generating matrix $\Delta\bSigma^{-1/2}$ and let $\m{V}=\m{V}(\Delta\bSigma^{-1/2})$ be \revAdd{(the interior of)} its Voronoi region. By~\eqref{eq:preVor}, we therefore have that
\begin{align}
\m{R}_{\text{MAP}}&=\left\{\bx\in\RR^K \ : \  \bSigma^{-1/2}\bx\in\m{V}(\Delta\bSigma^{-1/2}) \right\}.\label{eq:Rvor}
\end{align}
Noting that $\bZ\triangleq\bSigma^{-1/2}\bX\sim\m{N}(\mathbf{0},\bI_K)$, it therefore follows that
\begin{align}
\Pr(g_{\text{MAP}}(\bX^*)\neq\bX)=\Pr\left(\bZ\notin\m{V}(\Delta\bSigma^{-1/2})\right).\label{eq:PeMAP}
\end{align}

\subsection{Integer-Forcing Decoder}
\label{subsec:IFD}

In this section we recall a sub-optimal decoder, called the \emph{integer-forcing decoder}, which was introduced is~\cite{oe13b}, and is based on the simple observation from Proposition~\ref{prop:integermod}.

Let $\bA=[\ba_1|\cdots|\ba_K]^T\in\ZZ^{K\times K}$ be a full-rank integer matrix. By Proposition~\ref{prop:integermod}, we have that for any $\bx\in\RR^K$
\begin{align}
[\bA\bx^*]^*=[\bA\bx]^*.
\end{align}
Thus, \revAdd{$\bA\bx\in\CUBE$ is equivalent to $\bx=\bA^{-1}\left([\bA\bx^*]^*\right)$}. For $\bX\sim\m{N}(\mathbf{0},\bSigma)$, we therefore have that
\begin{align}
\Pr\left(\bA^{-1}\left([\bA\bX^*]^*\neq\bX\right)\right)&=\Pr(\bA\bX\notin\CUBE)\\
&=\Pr\left(\bigcup_{k=1}^K\left\{|\ba_k^T\bX|\geq \frac{\Delta}{2}\right\}\right).
\end{align}
Applying the union bound, this gives
{\small
\begin{equation}
	\begin{split}
		2\cdot Q\left(\frac{\Delta}{2\max_{k\in[K]} \sqrt{\ba_k^T\bSigma\ba_k} }\right)
		&\leq\Pr\left(\bA\bX\notin\CUBE\right) \\
		&\leq 2K\cdot  Q\left(\frac{\Delta}{2\max_{k\in[K]} \sqrt{\ba_k^T\bSigma\ba_k} }\right).
	\end{split}
\end{equation}}
The integer-forcing decoder solves\footnote{In fact, the restriction $\bA\in\SLk$ can be relaxed to $|\bA|\neq 0$ in the informed case, and was defined this way in~\cite{oe13b}. In the blind case, relaxing in~\eqref{eq:step2} the constraint $\bA\in\SLk$ to $|\bA|\neq 0$, may be problematic. As in practice $\bA$ is typically found using a lattice reduction algorithm such as LLL, which always returns $\bA\in\SLk$, we do not view this as a significant limitation.}
\begin{align}
\bA&=[\ba_1|\cdots|\ba_K]^T=\argmin_{\bar{\bA}\in\SLk}\max_{k=1,\ldots,K} \bar{\ba}_k^T\bSigma\bar{\ba}_k,\label{eq:Aopt}
\end{align} 
and estimates $\bX$ from $\bX^*$, by computing
\begin{align}
\hat{\bX}_{\text{IF}}\triangleq \bA^{-1}\left([\bA\bX^*]^*\right).
\end{align} 
Let 
\begin{align}
\sigma^2_K(\bSigma^{1/2})\triangleq \min_{\bar{\bA}\in\SLk}\max_{k\in[K]} \bar{\ba}_k^T\bSigma\bar{\ba}_k,\label{eq:MaxVarUni}
\end{align}
Thus, choosing $\bA$ as in~\eqref{eq:Aopt}, we have that
{\small \begin{align}
2\cdot Q\left(\frac{\Delta}{2\sigma_K(\bSigma^{1/2}) }\right)\leq\Pr\left(\hat{\bX}_{\text{IF}}\neq\bX\right)\leq 2K\cdot  Q\left(\frac{\Delta}{2\sigma_K(\bSigma^{1/2}) }\right),\label{eq:IFPeBounds}
\end{align}}
where the lower bound holds for any choice of $\bA\in\SLk$. If $\bSigma$ satisfies assumption $A1$, applying~\eqref{eq:IFPeBounds} yields
\begin{align}
2\cdot Q\left(\frac{\Delta}{2\sigma_K(\bSigma^{1/2}) }\right)\leq \min_{\bA\in\SLk}\Pr(\bA\bX\notin\CUBE)\leq \epsilon.
\end{align}
Consequently, if $\bSigma$ satisfies assumption $A1$, we have that
\begin{align}
\sigma_K(\bSigma^{1/2})\leq\frac{\Delta}{2}\cdot\frac{1}{Q^{-1}\left(\frac{\epsilon}{2}\right)} .\label{eq:A1implication}
\end{align}
\revAdd{Thus, if the error probability of the informed integer forcing decoder is small, then the lattice $\bSigma^{1/2}\ZZ^K$ has a basis consisting of relatively short vectors. This fact will be important in the analysis of our algorithm in the sequel.}

\section{Numerical Results}
\label{sec:sims}

Before turning to the technical analysis of the proposed recovery algorithm, we provide some numerical results that demonstrate its strength. The purpose of the experiments is to compare the performance of the proposed blind algorithm, with those of an ``informed'' benchmark algorithm that has access to $\bSigma$, under various assumptions on the structure of $\bSigma$. As the proposed blind algorithm essentially tries to mimic the (informed) \emph{integer-forcing decoder}, described in the Section~\ref{subsec:IFD}, we will naturally choose the latter as our benchmark.

In all experiments below, we consider $K\times K$ covariance matrices of the form
\begin{align}
\bSigma=\bI_{K}+\mathrm{snr}\bH\bH^T,\label{eq:covModel}
\end{align}
where $\bH\in\RR^{K\times \mathrm{rank}}$ for some integer $\mathrm{rank}$, and $\mathrm{snr}>0$ is a parameter. Such covariance matrices correspond to the output of a narrowband MIMO channel $\bY_t=\bH\bS_t+\bZ_t$, where $\bH\in\RR^{K\times \mathrm{rank}}$ is the channel matrix, $\bZ_t~\sim\m{N}(\mathbf{0},a\bI_K)$ is additive white Gaussian noise, and $\bS_t\sim\m{N}(\mathbf{0},b\bI_{\mathrm{rank}})$ is the vector of communication symbols transmitted at time $t$, which we model as Gaussian i.i.d, implicitly assuming underlying Gaussian  codebooks were used. If one then applies a modulo ADC with modulo size $\Delta$ on the output of each receive antenna, the resulting $K$-dimensional vector would be of the form $\bX_t^*$, where $\bX_t=\bY_t+\bQ_t$, and $\bQ_t$ is the vector of quantization noises incurred by the modulo ADCs. Further making the simplifying assumption that $\bQ_t\sim\m{N}(\mathbf{0},c\bI_K)$ and is statistically independent\footnote{\revAdd{Using dithered quantization, we get i.i.d. quantization error $\bQ$, albeit uniform over a quantization cell, rather than Gaussian. For the high-SNR regime examined in the numerical experiments here, this difference should have a negligible effect on performance.}} of $\bX_t$, we obtain that up to scaling, $\bX_t\sim\m{N}(\mathbf{0},\bSigma)$, where $\mathrm{snr}=b/(a+c)$. Thus, if the proposed blind algorithm succeeds with high probability in blindly recovering $\{\bX_1,\ldots,\bX_n\}$ from $\{\bX^*_1,\ldots,\bX^*_n\}$, then modulo-ADCs combined with this blind recovery algorithm can be effectively used for the corresponding MIMO channel, i.e., the effect of modulo reduction can be (blindly) undone with high probability.

We will evaluate the performance of the algorithm for various values of $K$, $\mathrm{rank}$, and $\mathrm{snr}$. We find it convenient to set $\Delta=\Delta(\bSigma,\epsilon)$ such that the error probability lower bound~\eqref{eq:IFPeBounds} of the informed benchmark integer-forcing decoder is fixed to some value, say $\epsilon$. In particular, $\Delta(\bSigma,\epsilon)$ is chosen by solving~\eqref{eq:MaxVarUni}, and setting 
\begin{align}
\Delta(\bSigma,\epsilon)=2\sigma_K(\bSigma^{1/2})Q^{-1}\left(\frac{\epsilon}{2}\right),
\end{align}
such that by~\eqref{eq:IFPeBounds} we have that $\Pr\left(\hat{\bX}_{\text{IF}}\neq\bX \right)\geq \epsilon$. The recovery error probability for the informed integer-forcing decoder is therefore lower bounded by
\begin{align}
\Pr\left(\{\hat{\bX}_{\text{IF},1},\ldots,\hat{\bX}_{\text{IF},n}\}\neq\{\bX_1,\ldots,\bX_n\}\right)\geq 1-(1-\epsilon)^n.\label{eq:IFsimLB}
\end{align}

For covariance matrices of the form~\eqref{eq:covModel}, we have that all eigenvalues of $\bSigma$ are at least $1$, with equality whenever $\mathrm{rank}<K$. Thus, we have that assumption A2 holds with $\tau_{\text{min}}=1$, for all covariance matrices we consider here. Furthermore, by our choice of $\Delta(\bSigma,\epsilon)$ and the upper bound from~\eqref{eq:IFPeBounds}, we have that assumption A1 holds with parameter $K\epsilon$. In section~\ref{sec:TruncEst}, we will see that a good choice for the parameter $d$ of the algorithm is of the form
\begin{align}
d=\eta\cdot\sqrt{K}\cdot \tau_{\text{min}}\cdot Q^{-1}\left(\frac{\epsilon}{2}\right),\label{eq:dchoice}
\end{align}
where $\eta$ is a parameter that trades-off between the error probability and the number of samples $n$ the algorithm requires in order to succeed. In all our simulations we have chosen $d$ as in~\eqref{eq:dchoice}, with $\eta=1/2$ and $\tau_{\text{min}}=1$.\footnote{The analysis in Section~\ref{sec:TruncEst} shows that if $\bV\sim\m{N}(\mathbf{0},\bR)$, and $R$ satisfies assumptions $A1$ and $A2$, choosing $d$ as in~\eqref{eq:dchoice} enables us to control the probability $p_{\text{est-err}}$ of inaccurate estimation of $\mathbb{E}[\bV\bV^T|\bV\in\CUBE]$. Note, however, that the fact that $\bSigma$ satisfies assumption $A1$ with parameter $\tau_{\text{min}}$, does not guarantee that $\bR=\bA\bSigma\bA^T$ will have smallest eigenvalue lower bounded by $\tau^2_{\text{min}}$. Thus, choosing $d$ as in~\eqref{eq:dchoice} only guarantees $p_{\text{est-err}}$ is controlled for the first iteration. Nevertheless, intuitively we expect $\bA\bX$ to become ``whiter'' from iteration to iteration, such that the minimal eigenvalue of $\bA\bSigma\bA^T$ is expected to grow. For this reason, in the simulations we have used~\eqref{eq:dchoice} throughout all iterations. For the proof of Theorem~\ref{thm:success_estimator} and Theorem~\ref{thm:success_recovery} we take a more conservative approach and replace $\tau_{\text{min}}$ in~\eqref{eq:dchoice} with $1/\chi_1(\epsilon;P,K)$, which is shown to universally lower bound the smallest squared-root eigenvalue of $\bA\bSigma\bA^T$ throughout all iterations of the algorithm.} We have also set the design parameter $M$ controlling the maximum number of iterations the algorithm performs to $30$.

For fixed $\bH$, $\mathrm{snr}$, $\epsilon$, and $n$, let $\bSigma$ be as in~\eqref{eq:covModel}  and define
\begin{align}
P_e^\text{blind}(\bH,\mathrm{snr},\epsilon,n)\triangleq \Pr\left(\{\hat{\bX}_{1},\ldots,\hat{\bX}_{n}\}\neq\{\bX_1,\ldots,\bX_n\}\right)
\end{align}
where $\{\hat{\bX}_1,\ldots,\hat{\bX}_n\}$ are the estimates for $\{\bX_1,\ldots,\bX_n\}$ our proposed algorithm (with the choice of $d$ and $M$ as above) computes from $\{\bX^*_1,\ldots,\bX^*_n\}$, where the modulo size is $\Delta(\bSigma,\epsilon)$. Furthermore, for fixed $K$ and $\mathrm{rank}$ we assume the entries of $\bH$ are i.i.d. $\m{N}(0,1)$, and define
\begin{align}
P_e^{\text{blind}}(K,\mathrm{rank},\mathrm{snr},\epsilon,n)\triangleq \mathbb{E}\left[P_e^\text{blind}(\bH,\mathrm{snr},\epsilon,n)\right],
\end{align}
where the expectation is with respect to the randomness in $\bH$. Since exact computation of $P_e^{\text{blind}}(K,\mathrm{rank},\mathrm{snr},\epsilon,n)$ is not feasible, we estimate it using Monte-Carlo simulations. In particular, for fixed $K,\mathrm{rank},\mathrm{snr},\epsilon$ and $n$ we draw $r=1000$ different realizations of $\bH$, and for each one of those compute $\bSigma$ according to~\eqref{eq:covModel}. We then draw $\bX_1,\ldots,\bX_n\stackrel{i.i.d.}{\sim}\m{N}(\mathbf{0},\bSigma)$, compute $\{\bX^*_1,\ldots,\bX^*_n\}$, where $\Delta=\Delta(\bSigma,\epsilon)$, and apply the proposed algorithm that yield estimates $\{\hat{\bX}_1,\ldots,\hat{\bX}_n\}$. Our estimate for $P_e^{\text{blind}}(K,\mathrm{rank},\mathrm{snr},\epsilon,n)$ is taken as the fraction of realizations of $\bH$ for which $\{\hat{\bX}_1,\ldots,\hat{\bX}_n\}\neq\{\bX_1,\ldots,\bX_n\} $

Figure~\ref{fig:Gap13} shows the Monte-Carlo estimate of $P_e^{\text{blind}}(K,\mathrm{rank},\mathrm{snr},\epsilon,n)$ for $n=1000$, $\epsilon=2\cdot 10^{-5}$ and various values of $K$, $\mathrm{rank}$ and $\mathrm{snr}$, whereas in Figure~\ref{fig:Gap15} we took $n=1000$ as well, but with $\epsilon\approx 9.6\cdot 10^{-7}$. In all experiments, whenever we had to solve an integer optimization problem of the form of~\eqref{eq:Aopt}, we have used the LLL algorithm~\cite{lll82} in order to get an approximate, possibly sub-optimal, solution. The numerical results demonstrate several trends.
\begin{enumerate}
	\item For a fixed $n$ and $\epsilon$, the analysis in the proceeding sections predicts that performance should severely degrade as the problem dimension $K$ increases. The reason for this is that the success of the $\mathrm{EstimateTruncatedCovariance}$ algorithm depends on whether or not the points $\{\bX_1,\ldots,\bX_n\}$ form a dense covering of some typical set of the $\m{N}(\mathbf{0},\bSigma)$ distribution, which becomes more difficult as $K$ increases. In contrast, the numerical results show that what has the greater effect on performance is $\min\{K,\mathrm{rank}\}$, rather than $K$. This is quite intuitive, as in the limit of high-SNR, the distribution of $\bX$ is essentially supported on a subspace of dimension $\min\{K,\mathrm{rank}\}$. 
	
	From a practical point of view, this observation is very encouraging, as in many scenarios of interest the source is of a sparse nature. For example, in massive MIMO the number of receive antennas $K$ may be very large, but the number of transmitters, corresponding to $\mathrm{rank}$ in our notation, is relatively small.
	\item For small rank and reasonable $\mathrm{snr}$, a moderate number of measurements ($n=1000$) suffices for the proposed algorithm to achieve error probability which is close to that of the (informed) benchmark recovery algorithm, even for $K$ as large as $10$.
	\item For fixed $n$, the error probability of the proposed algorithm tends to become closer to that of the benchmark, as $\epsilon$ decreased. More precisely, the range of parameters $K$, $\mathrm{\rank}$ and $\mathrm{snr}$ for which $\frac{P_e^{\text{blind}}(K,\mathrm{rank},\mathrm{snr},\epsilon,n)}{1-(1-\epsilon)^n}<c$, for some $c>1$, tends to increase with $\epsilon$. This effect also comes up in our analysis.
	\item The proposed algorithm performs better at low-SNR, i.e., when the source $\bX$ is less correlated. The reason for this is that in this case the initial conditions for the algorithm are more favorable. A large fraction of $\{\bX_1,\ldots,\bX_n\}$ are in $\CUBE$ to begin with, and less ``work'' is required in order to manipulate the distribution, via multiplication by integer-matrices, to one where a random point falls within $\CUBE$ with high probability. This phenomenon will also be evident in our analysis.
\end{enumerate}

\begin{figure*}[]
	\begin{center}
		\subfloat[$\mathrm{rank}=1$]{
			\includegraphics[width=0.45\textwidth]{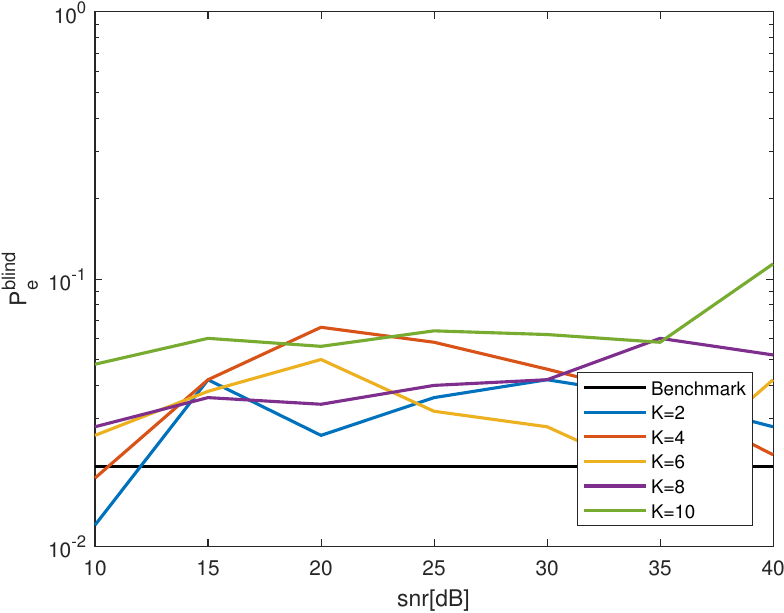}
			\label{fig:100Hz}}
		\qquad
		\subfloat[$\mathrm{rank}=2$]{
			\includegraphics[width=0.45\textwidth]{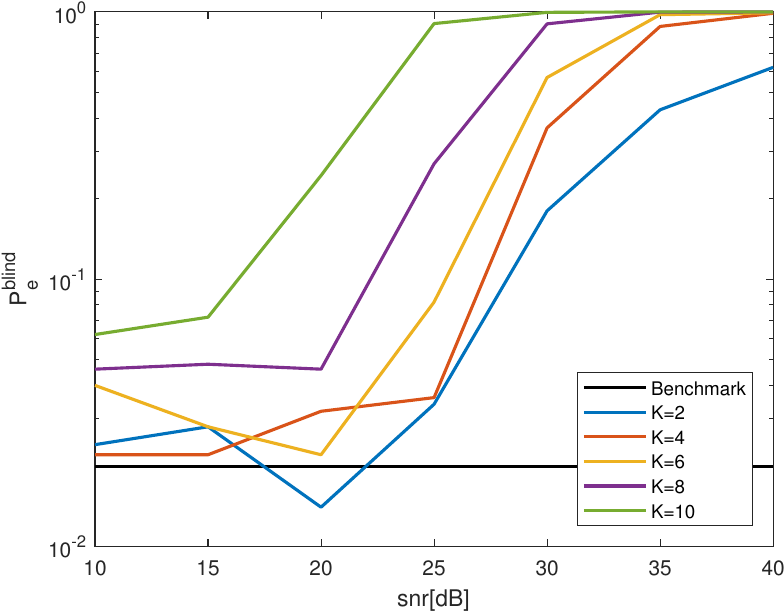}
			\label{fig:44KHz}}
		
		\subfloat[$\mathrm{rank}=3$]{
			\includegraphics[width=0.45\textwidth]{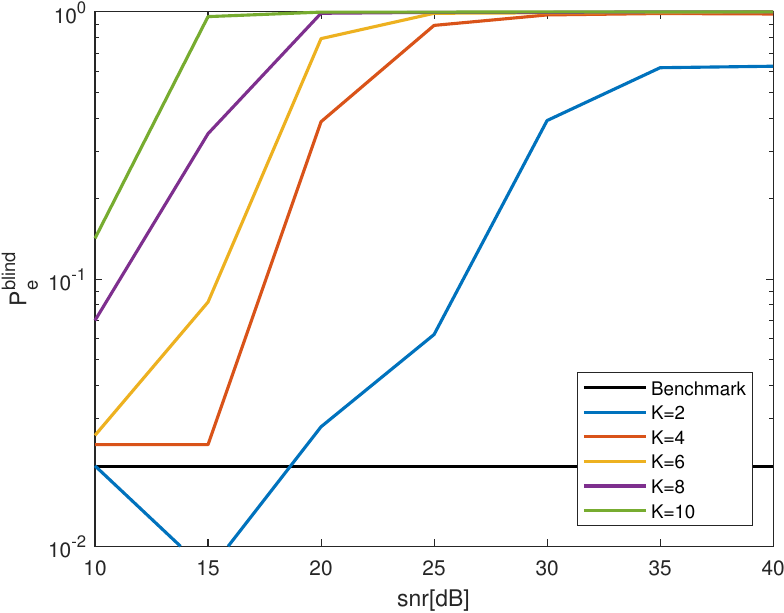}
			\label{fig:100KHz}}
		\qquad
		\subfloat[$\mathrm{rank}=4$]{
			\includegraphics[width=0.45\textwidth]{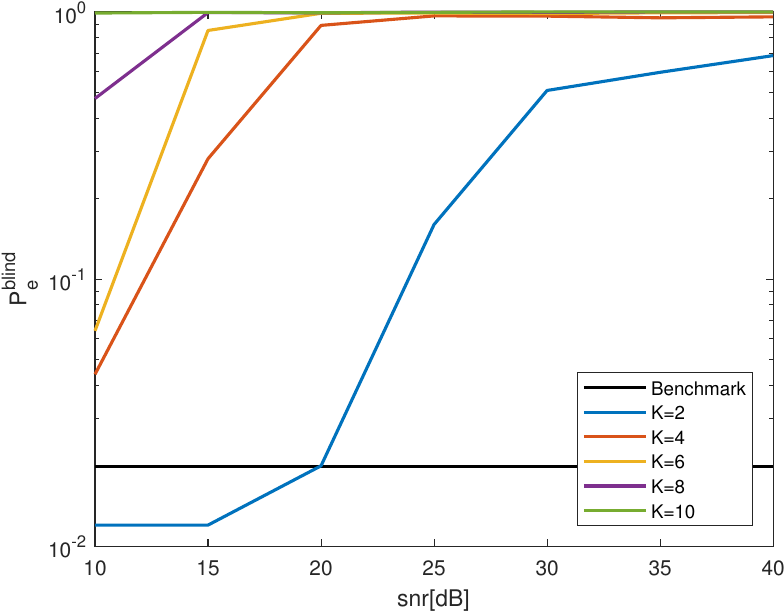}
			\label{fig:1MHz}}
	\end{center}
	\caption{Monte-Carlo evaluation of $P_e^{\text{blind}}(K,\mathrm{rank},\mathrm{snr},\epsilon,n)$ for $n=1000$, $\epsilon=2\cdot 10^{-5}$, such that $1-(1-\epsilon)^n=0.0198$, and various values of $K$, $\mathrm{rank}$ and $\mathrm{snr}$.}
	\label{fig:Gap13}
\end{figure*}

\begin{figure*}[]
	\begin{center}
		\subfloat[$\mathrm{rank}=1$]{
			\includegraphics[width=0.45\textwidth]{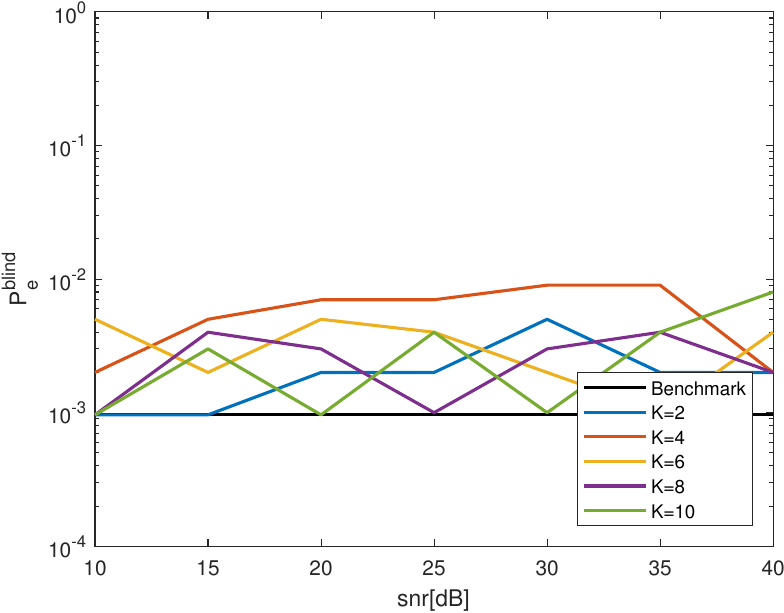}
			\label{fig:100Hz}}
		\qquad
		\subfloat[$\mathrm{rank}=2$]{
			\includegraphics[width=0.45\textwidth]{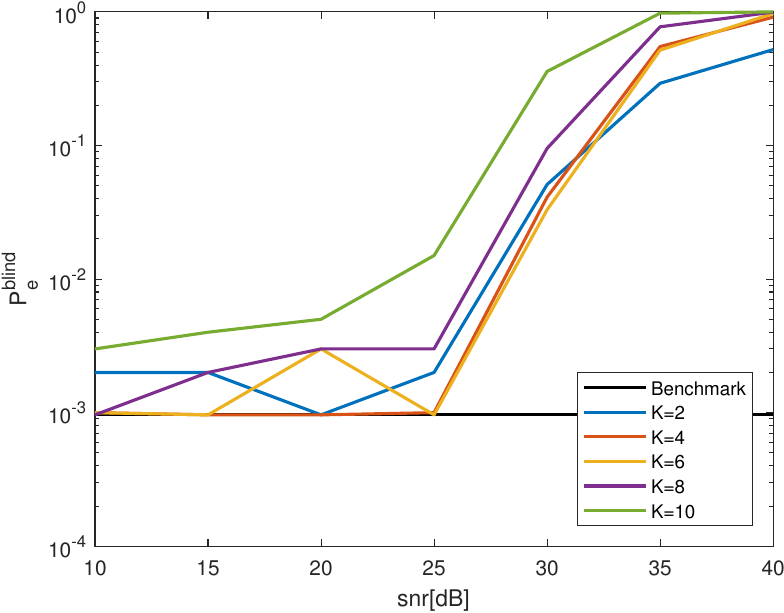}
			\label{fig:44KHz}}
		
		\subfloat[$\mathrm{rank}=3$]{
			\includegraphics[width=0.45\textwidth]{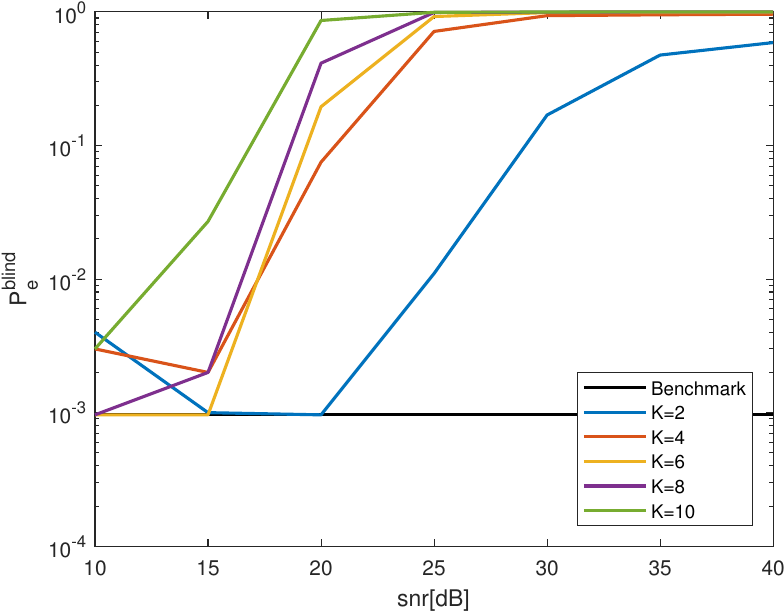}
			\label{fig:100KHz}}
		\qquad
		\subfloat[$\mathrm{rank}=4$]{
			\includegraphics[width=0.45\textwidth]{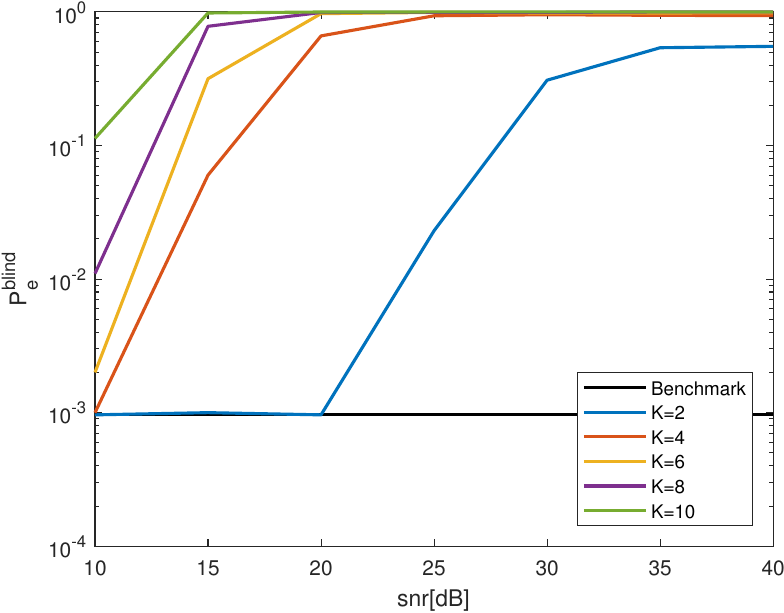}
			\label{fig:1MHz}}
	\end{center}
	\caption{Monte-Carlo evaluation of $P_e^{\text{blind}}(K,\mathrm{rank},\mathrm{snr},\epsilon,n)$ for $n=1000$, $\epsilon\approx 9.6\cdot 10^{-7}$, such that $1-(1-\epsilon)^n\approx 9.6\cdot 10^{-4}$, and various values of $K$, $\mathrm{rank}$ and $\mathrm{snr}$.}
	\label{fig:Gap15}
\end{figure*}

\section{Analysis of Genie-Aided Algorithm}
\label{sec:genieaided}

The purpose of this section is to prove Theorem~\ref{thm:genie}. As the proof involves quite a few steps, we begin this section with a high-level overview. For this high-level description only, we assume that $\beta=0$, i.e., that the genie replacing the procedure $\mathrm{EstimateTruncatedCovariance}$ returns a perfect estimate for the covariance matrix of the truncated random vector. The complete analysis takes into account the effect of finite $\beta$. The main ideas in our analysis are as follows:

\begin{enumerate}
\item First, we show that the covariance matrix of a Gaussian random vector $\bX\sim\m{N}(\mathbf{0},\bSigma)$ truncated to a convex symmetric body $\m{S}$, as is $\CUBE$, satisfies 
\begin{align}
\psi(\Pr(\bX\notin\m{S}))\cdot \bSigma\preceq\mathbb{E}[\bX\bX^T|\bX\in\m{S}]\preceq\bSigma,\label{eq:trunchighlevel}
\end{align} 
for some decreasing mapping $\psi:[0,1]\to[0,1]$ which we explicitly specify.
\item Let $\tilde{\bA}=\argmin_{\tilde{\bA}\in\SLk}\max_{k\in[K]}\tilde{\ba}^T_k \mathbb{E}[\bX\bX^T|\bX\in\CUBE]\tilde{\ba}_k$ be the matrix found in step $2$ of the main algorithm.
The upper bound from~\eqref{eq:trunchighlevel} implies that 
{\small\begin{align}
\max_{k\in[K]}\tilde{\ba}^T_k \mathbb{E}[\bX\bX^T|\bX\in\CUBE]\tilde{\ba}_k&\leq \min_{\bar{\bA}\in\SLk}\max_{k\in[K]} \bar{\ba}_k^T\bSigma\bar{\ba}_k\nonumber\\
&\leq \frac{\Delta}{2}\frac{1}{Q^{-1}\left(\frac{\epsilon}{2}\right)},\label{eq:varhighlevel}
\end{align}}
where the last inequality follows from~\eqref{eq:A1implication}. Furthermore, we can lower bound the left-hand side of~\eqref{eq:varhighlevel} using the lower bound from~\eqref{eq:trunchighlevel}, which gives
\revAdd{\footnotesize
\begin{equation}
	\begin{split}
		\max_{k\in[K]}\tilde{\ba}^T_k \bSigma\tilde{\ba}_k
		&\leq
		\frac{1}{\psi(\Pr(\bX\notin\CUBE))}\max_{k\in[K]}\tilde{\ba}^T_k \mathbb{E}[\bX\bX^T|\bX\in\CUBE] \tilde{\ba}_k\\
		&\leq
		\frac{1}{\psi(\Pr(\bX\notin\CUBE))}\frac{\Delta}{2}\frac{1}{Q^{-1}\left(\frac{\epsilon}{2}\right)}.\label{eq:varhighlevelsmallprob}
	\end{split}
\end{equation}}
\item The bound~\eqref{eq:varhighlevelsmallprob}, implies that if $P=\Pr(\bX\notin\CUBE)$ is not too large, such that $\psi(P)$ is not too small, then $\Pr(\tilde{\bA}\bX\in\CUBE)$ is not much greater than $\epsilon$. It follows that once the algorithm reaches a point where the probability of missing $\CUBE$ is moderate, say below $1/30$, a few additional iterations will bring the probability of missing $\CUBE$ to almost as low as $\epsilon$. However, for $P$ close to $1$ we have that $\psi(P)$ is very small, and in this case~\eqref{eq:varhighlevelsmallprob} does not guarantee that the probability of missing $\CUBE$ decreases after an iteration of the algorithm.

\item Thus, a different technique, which is effective for $P$ close to $1$, is needed for showing that the probability of missing $\CUBE$ decreases.
By Chebyshev's inequality\footnote{\revAdd{The careful reader might note that, conditioned on $\bX\in \CUBE$, the random variable $\tilde{\ba}_k\bX$ is in fact bounded, hence in particular sub-Gaussian, so that Chebyshev's inequaltity should give, at least in principle, rather bad estimates on its tail. However, note that the bound on the support is $\left| \tilde{\ba}_k\bX \right| \le \frac\Delta2\|\tilde{\ba}_k\|_1$, and we do not have means of controlling the norm of $\tilde{\ba}_k$; this makes, e.g, Hoeffding's inequality inapplicable. Nonetheless, it turns out that this rather crude application of Chebyshev's inequality is already good enough to bootstrap our convergence argument.}}
 and~\eqref{eq:varhighlevel} we have that
\begin{align}
\max_{k\in[K]}\Pr\left(|\tilde{\ba}_k\bX|\geq \frac{1}{2}\cdot\frac{\Delta}{2} \ \bigg| \ \bX\in\CUBE \right)\leq \frac{4}{\left(Q^{-1}\left(\frac{\epsilon}{2}\right)\right)^2},
\end{align}
and by the union bound
\begin{align}
\Pr\left(\tilde{\bA}\bX\notin\frac{1}{2}\CUBE \ \big| \ \bX\in\CUBE\right)\leq \frac{4K}{\left(Q^{-1}\left(\frac{\epsilon}{2}\right)\right)^2}=
\delta.
\end{align}
Consequently,
{\small
\begin{align}
\Pr\left(\tilde{\bA}\bX\in\frac{1}{2}\CUBE\right)&\geq \Pr(\bX\in\CUBE)\nonumber\\
&\cdot \Pr\left(\tilde{\bA}\bX\in\frac{1}{2}\CUBE \ \big| \ \bX\in\CUBE\right)\nonumber\\
&\geq (1-\delta)\cdot\Pr(\bX\in\CUBE).\label{eq:highlevhalfcube}
\end{align}
}
Next, we leverage a result by Lata\l{}a and Oleszkiewicz~\cite{lo99} to derive a lower bound on $\Pr\left(\tilde{\bA}\bX\in\CUBE\right)$ as a function of $\Pr\left(\tilde{\bA}\bX\in\frac{1}{2}\CUBE\right)$, which holds universally for all $\bSigma$ and $\tilde{\bA}$. Combining this bound with~\eqref{eq:highlevhalfcube}, we obtain that
{\footnotesize\begin{align}
\Pr\left(\tilde{\bA}\bX\notin\CUBE\right)\leq 2Q\left(2Q^{-1}\left(\frac{1-(1-\delta)(1-\Pr(\bX\notin\CUBE))}{2}\right) \right).
\end{align} 
}
Thus, if $\delta$ is sufficiently small with respect to $\Pr(\bX\notin\CUBE)$, one iteration of the main algorithm will decrease the probability that the new random variable $\tilde{\bA}\bX$ falls outside of $\CUBE$. Applying the same argument again, shows that another iteration will further decrease the probability of falling outside of $\CUBE$, and so on.
\item Analyzing the dynamics of the process $P_n=2Q(2Q^{-1}\left(\frac{1-(1-\delta)(1-P_{n-1})}{2}\right))$, initialized with $P_0=P<1$, we see that after a given number of iterations we must have that the probability of falling outside of $\CUBE$ is close to $p^*$, where $p^*$ is the fixed-point of the dynamics. We show that if $Q^{-1}\left(\frac{\epsilon}{2}\right)<6\sqrt{K}$, then $p^*<1/30$. Consequently, after a large enough number of iterations (which we explicitly upper bound) we must have that the probability of missing $\CUBE$ is below $1/30$, and then, by~\eqref{eq:highlevhalfcube}, a few more iterations must decrease the probability of missing $\CUBE$ almost all the way to $\epsilon$.
\end{enumerate}

We now turn to the details. The following easy result will be useful throughout this paper.
\begin{proposition}
	\label{prop:contained_ball}
	Let $\m{S} \subset \mathbb{R}^K$ be a convex, symmetric body with $\mu(\m{S})\ge 1-\epsilon$, where $\mu$ is standard Gaussian measure. Then $\m{B}(\mathbf{0},r_0) \subset \m{S}$, with
	\begin{equation}
	r_0 = Q^{-1}\left(\frac{\epsilon}{2}\right) \,.
	\end{equation}
\end{proposition}
\begin{proof}
	This result is well-known (see, for example, (\cite{bogachev1998gaussian}, Proposition 4.10.11)), and we shall repeat its proof for completeness. Let $\bv\in\m{S}$ be a point with minimal norm on the boundary of $\m{S}$, and $H$ be a supporting hyperplane of $\m{S}$ tangent at $\bv$. Since $\m{S}$ is symmetric, $-\bv\in\m{S}$ and there is a supporting hyperplane $H^{-}$ tangent at $-\bv$ and parallel to $H$. Of course, $\text{dist}(H,H^{-})=2\norm{\bv}$, and letting $\m{R}$ be the slab enclosed by $H$ and $H^{-}$, we have 
	\[
	\m{S} \subset \m{R} \,,
	\]
	hence
	\[
	\epsilon \ge 1-\mu(\m{S}) \ge 1-\mu(\m{R}) = 2Q\left( \norm{\bv} \right) \,,
	\]
	so $\norm{\bv} \ge Q^{-1}\left(\frac {\epsilon}{2}\right)$.
\end{proof}

We start by deriving a basic property of the truncated Gaussian distribution.

\begin{lemma}
	Let $\bX\sim\m{N}(\mathbf{0},\bSigma)$, $\m{S}\in\RR^K$ be a convex symmetric set, and suppose that $\Pr(\bX\notin\m{S}) \le P$. Then
	\begin{align}
	\frac{\alpha(P;K)}{1-P}\bSigma\preceq\mathbb{E}[\bX\bX^T|\bX\in\m{S}]\preceq\bSigma,
	\end{align}
	where
	\begin{align}
	\alpha(P;K)\triangleq \max \left[ F_{\chi^2(K+2)}\left( \left( Q^{-1}\left(\frac{P}{2}\right)\right)^2 \right), 1-\sqrt{3P}\right] \,,
	\label{eq:alphadef}
	\end{align}
	and $F_{\chi^2(K+2)}$ is the CDF of a (standard) $\chi^2$ random variable with $K+2$ degrees of freedom. \footnote{ That is, of the random variable $W=\sum_{i=1}^{K+2}Z_i^2$, where the $Z_i$-s are $K+2$ i.i.d standard Gaussians $\sim \m{N}(0,1)$.}
	\label{lemma:truncvar}
\end{lemma}

\begin{proof}
We begin by proving the upper bound. By the Gaussian Correlation Inequality (see \cite{royen2014simple} for the proof, and also \cite{latala2017royen} for a more expository account), for every $f,g:\RR^K \to \RR$ symmetric and quasiconcave \footnote{Namely, such that for all $0\le \lambda \le 1$, $f(\lambda\bx +(1-\lambda\by) \ge \min\left(f(\bx),f(\by)\right)$.} it holds that
\[
	\E{ f(\bX)\cdot g(\bX) } \ge \E{f(\bX)} \cdot \E{g(\bX)} \,.
\]
Since $\m{S}$ is symmetric convex, the indicator $\bx \mapsto \Ind_{\bx\in\m{S}}$ is symmetric quasiconcave, and clearly for every $\bu\in\RR^K$, $\bx\mapsto -(\bu^T\bx)^2$ is concave (hence quasiconcave). Applying the Gaussian Correlation Inequality,
	\begin{align*}
		\E{ (\bu^T\bX)^2 \cdot \Ind_{\bX\in\m{S}} } 
		&\le \E{(\bu^T\bX)^2} \cdot \E{ \Ind_{\bX\in\m{S}}} \\
		&= \bu^T\bSigma\bu \cdot \Pr(\bX\in\m{S}) \,.
	\end{align*}
	Dividing by $\Pr(\bX\in\m{S})$ now gives the required inequality $\E{(\bu^T\bX)^2|\bX\in\m{S}} \le \bu^T\bSigma\bu$.

For the lower bound, let $\bZ=\bSigma^{-1/2}\bX\sim\m{N}(\mathbf{0},\bI_K)$, and denote $\m{T}=\bSigma^{-1/2}\m{S}$. Then
\begin{align*}
\E{ \bX\bX^T | \bX\in\m{S} } 
&= \bSigma^{1/2}\E{ \bZ\bZ^T | \bZ\in\m{T} } \bSigma^{1/2} \\
&\ge \bSigma^{1/2} \cdot \frac{1}{1-P} \cdot \E{ \bZ\bZ^T \cdot \Ind_{\bZ\in\m{T}} } \cdot \bSigma^{1/2} \,,
\end{align*}
where we have used the fact that $\Pr(\bX\in\m{S})=\Pr(\bZ\in\m{T}) \ge 1-P$. Since $\m{S}$ is convex symmetric, so is $\m{T}$, and therefore, by Proposition \ref{prop:contained_ball}, $\m{T}$ contains the ball $\m{B}\left(\mathbf{0},Q^{-1}\left(\frac{P}{2}\right)\right)$. For every $\bu\in\RR^K$ with $\norm{\bu}=1$ we clearly have
\begin{align*}
\E{ \left(\bu^T \bZ\right)^2 \cdot \Ind_{\bZ\in\m{T}} } 
&\ge \E{ \left(\bu^T \bZ\right)^2 \cdot \Ind_{\bZ\in\m{B}\left(\mathbf{0},Q^{-1}\left(\frac{P}{2}\right)\right)} } \,.
\end{align*}
We now calculate the expression on the right. By the rotational invariance of the Gaussian distribution, we may assume without loss of generality that $\bu=\be_i$, that is,    
\begin{align*}
	\E{ \left(\be_i^T \bZ\right)^2 \cdot \Ind_{\bZ\in\m{B}\left(\mathbf{0},Q^{-1}\left(\frac{P}{2}\right)\right)} } 
	&= \E{ Z_i^2 \cdot \Ind_{\bZ\in\m{B}\left(\mathbf{0},Q^{-1}\left(\frac{P}{2}\right)\right)} } \\
	&= \frac{1}{K} \sum_{i=1}^K \E{ Z_i^2 \cdot \Ind_{\bZ\in\m{B}\left(\mathbf{0},Q^{-1}\left(\frac{P}{2}\right)\right)} } \\
	&= \frac{1}{K} \E { \norm{\bZ}^2 \cdot \Ind_{\norm{Z}\le Q^{-1}\left(\frac{P}{2}\right) }} \,.
\end{align*}         
Moving to polar coordinates gives,
{\small \begin{align*}
	\frac{1}{K} \E { \norm{\bZ}^2 \Ind_{\norm{Z}\le d} } 
	&= \frac{1}{K} (2\pi)^{-K/2} \int_{\m{B}(0,d) } \norm{\bZ}^2 e^{-\frac{1}{2}\norm{\bZ}^2 } d\bZ \\
	&= \frac{1}{K} (2\pi)^{-K/2}   \Vol_{K-1}\left(S^{K-1}\right) \int_{0}^{d} r^{K+1} e^{-r^2/2} dr \,.
\end{align*}}
Recalling that $\Vol_{K-1}\left(S^{K-1}\right) = K \cdot V_K = K\cdot \frac{\pi^{K/2}}{\Gamma\left(\frac{K}{2}+1\right)}$, and performing a change-of-variables $u=r^2$ in the integral, we are left with
\[
\frac{1}{2^{\frac{K}{2}+1} \cdot \Gamma(\frac{K}{2}+1)} \cdot \int_{0}^{d^2} u^{\frac{K}{2}}e^{-u/2}du = F_{\chi^2(K+2)}(d^2) \,,
\]
as claimed. 

We also derive another simple lower bound on $\E{ \left(\bu^T \bZ\right)^2 \cdot \Ind_{\bZ\in\m{T}}}$ that is weaker for large $P$, and may in fact be negative, but is dimension independent and is therefore stronger for small $P$ and large $K$. 
Our bound follows from a simple application of the Cauchy-Schwartz inequality,
\[
\E{ \left(\bu^T \bZ\right)^2 \cdot \Ind_{\bZ\notin\m{T}} } \le \sqrt{ \E{ \left(\bu^T \bZ\right)^4} } \cdot \sqrt{\Pr\left(\bZ\notin\m{T}\right)} \le \sqrt{3P} \,,
\]
where se have used the fact that $\bu^T\bZ \sim \m{N}(0,1)$ in the last step.
Hence,
\[
\E{ \left(\bu^T \bZ\right)^2 \cdot \Ind_{\bZ\in\m{T}} } \ge 1-\sqrt{3P} \,,
\]
which is non-trivial when $P<1/3$.
\end{proof}


Using the upper bound from Lemma~\ref{lemma:truncvar}, we can prove the following.

\begin{lemma}[Quality of integer solution at the $m$-th iteration]
Consider the (genie-aided) main algorithm after $m$ iterations have been preformed. Let $\bA=\bA^{(m)}$ be the integer matrix a this point, and define the random vector $\bV=\bV^{(m)}=\bA\bX$. Let $\tilde{\bA}=[\tilde{\ba}_1|\cdots|\tilde{\ba}_K]$ be the integer matrix found at step $2$ of the $m+1$th iteration of the main algorithm. Then
\begin{align}
\max_{k\in[K]} \tilde{\ba}_k^T\mathbb{E}\left[\bV\bV^T|\bV\in\CUBE\right]\tilde{\ba}_k\leq \frac{1+\beta}{1-\beta}\sigma^2_K(\bSigma^{1/2}),
\end{align}
where $\sigma^2_K(\bSigma^{1/2})$ is as defined in~\eqref{eq:MaxVarUni}.
\label{lem:intqual}
\end{lemma}

\begin{proof}
Let 
$\bA^{\text{opt}}=[\ba^{\text{opt}}_1|\cdots|\ba_K^{\text{opt}}]^T\in\SLk$ be the optimal integer-forcing matrix for $\bSigma$, as defined in~\eqref{eq:Aopt}, and recall that
\begin{align}
\sigma_K(\bSigma^{1/2})\triangleq \max_{k\in[K]}\sqrt{\ba^{\text{opt},T}_k\bSigma\ba^{\text{opt}}_k}.\label{eq:sigmaKest2}
\end{align}
Let $\tilde{\bA}_i$, $i=1,\ldots,m$ be the integer matrices found in each one of the first $m$ iterations, such that the current integer matrix is $\bA=\tilde{\bA}_m\cdot\ldots\cdot\tilde{\bA}_1$. 
Note that $\bA\in\SLk$, as it is the product of elements of the group $\SLk$. Thus, $\bA^{-1}$ is also in the group $\SLk$. It therefore follows that the matrix $\tilde{\bA}^{\text{opt}}=\bA^{-T}\bA^{\text{opt}}=[\tilde{\ba}_1^{\text{opt}}|\cdots|\tilde{\ba}_K^{\text{opt}}]^T$ is in $\SLk$, where $\tilde{\ba}_k^{\text{opt}}\triangleq\bA^{-T}\ba_k^{\text{opt}}\in\ZZ^K$, $k=1,\ldots,K$. 

Let $\bSigmaV$, be the estimated covariance matrix the genie produces at step $1$ of the $m+1$th iteration of the main algorithm, and let
\begin{align}
\tilde{\bA}&=[\tilde{\ba}_1|\cdots|\tilde{\ba}_K]^T=\argmin_{\bar{\bA}\in\SLk}\max_{k\in[K]} \bar{\ba}_k^T\overset{\vee}{\bSigma}\bar{\ba}_k,\label{eq:integeroptimization}
\end{align} 
be the matrix found in step $2$ of the $m+1$th iteration.
We can now write
{\small
	\begin{align}
\max_{k\in[K]} \tilde{\ba}_k^T\mathbb{E}&\left[\bV\bV^T|\bV\in\CUBE\right]\tilde{\ba}_k \nonumber \\
&\leq \frac{1}{1-\beta}\max_{k\in[K]} \tilde{\ba}_k^T\bSigmaV\tilde{\ba}_k\label{eq:estLB}\\
&\leq \frac{1}{1-\beta}\max_{k\in[K]}\tilde{\ba}_k^{\text{opt},T}\bSigmaV\tilde{\ba}_k^{\text{opt}}\label{eq:defoptimizer}\\
&\leq \frac{1+\beta}{1-\beta}\max_{k\in[K]} \tilde{\ba}_k^{\text{opt},T}\mathbb{E}\left[\bV\bV^T|\bV\in\CUBE\right]\tilde{\ba}_k^{\text{opt}}\label{eq:estUB}\\
&\leq \frac{1+\beta}{1-\beta}\max_{k\in[K]} \tilde{\ba}_k^{\text{opt},T}\bA\bSigma\bA^T\tilde{\ba}_k^{\text{opt}}\label{eq:truncUB}\\
&=\frac{1+\beta}{1-\beta}\sigma^2_K(\bSigma^{1/2}),\label{eq:finalvarineq}
\end{align}
}
where~\eqref{eq:estLB} follows from the lower bound in~\eqref{eq:estassumption},~\eqref{eq:defoptimizer} follows from the definition of $\tilde{\bA}$ in~\eqref{eq:integeroptimization},~\eqref{eq:estUB} follows from the upper bound in~\eqref{eq:estassumption},~\eqref{eq:truncUB} follows from the upper bound in Lemma~\ref{lemma:truncvar} (recall that $\bV\sim\m{N}(\mathbf{0},\bA\bSigma\bA^T)$), and~\eqref{eq:finalvarineq} from~\eqref{eq:sigmaKest2}. 
Substituting~\eqref{eq:finalvarineq} into~\eqref{eq:simpdyn} yields the desired result.
\end{proof}

Next, we use the two preceding lemmas, in order to show that the probability of missing $\CUBE$ decreases from iteration to iteration in the (genie-aided) main algorithm. To this end we develop two bounds on the dynamics of this probability. Our first bound follows from a rather simple application of Lemma~\ref{lemma:truncvar} and Lemma~\ref{lem:intqual}, but is only useful for showing that if the probability of missing $\CUBE$ is already quite small, it will become smaller after another iteration. When the probability of missing $\CUBE$ is not small enough, a different technique is needed.

\begin{lemma}
	Consider the (genie-aided) main algorithm after $m$ iterations have been preformed. Let $\bA=\bA^{(m)}$ be the integer matrix at this point, define the random vector $\bV=\bV^{(m)}=\bA\bX$, and let $P_m\triangleq \Pr(\bV^{(m)}\notin\CUBE)$. Then
	\begin{align}
	P_{m+1}<K\cdot Q\left(\sqrt{\frac{1-\beta}{1+\beta}\cdot\frac
		{\alpha(P_m;K)}{1-P_m}}\cdot Q^{-1}\left(\frac{\epsilon}{2}\right)\right),
	\end{align}
	where $\alpha(P;K)$ is as defined in~\eqref{eq:alphadef}.
	\label{lem:fastconvergence}
\end{lemma}

\begin{proof}
Let $\tilde{\bA}=[\tilde{\ba}_1|\cdots|\tilde{\ba}_K]$ be the integer matrix found at step $2$ of the $m+1$th iteration of the main algorithm, and note that 
	\begin{align}
	P_{m+1}=\Pr(\tilde{\bA}\bV\notin\CUBE)\leq K\cdot Q \left(\frac{\Delta/2}{\sqrt{\max_{k\in[K]} \tilde{\ba}_k^T\bA\bSigma\bA\tilde{\ba}_k}}\right).\label{eq:simpdyn}
	\end{align}
	We upper bound $\tilde{\ba}_k^T\bA\bSigma\bA\tilde{\ba}_k$ as 
	\begin{align}
	\max_{k\in[K]} \tilde{\ba}_k^T\bA\bSigma\bA\tilde{\ba}_k&\leq \frac{1-P_m}{\alpha(P_m)}\max_{k\in[K]} \tilde{\ba}_k^T\mathbb{E}\left[\bV\bV^T|\bV\in\CUBE\right]\tilde{\ba}_k\label{eq:truncLB}\\
	&\leq\frac{1+\beta}{1-\beta}\frac{1-P_m}{\alpha(P_m)}\sigma^2_K(\bSigma^{1/2}),\label{eq:finalvarineq1}
	\end{align}
	where~\eqref{eq:truncLB} follows from the lower bound in Lemma~\ref{lemma:truncvar} and~\eqref{eq:finalvarineq1} from Lemma~\ref{lem:intqual}. 
	Substituting~\eqref{eq:finalvarineq1} into~\eqref{eq:simpdyn}, and recalling that $\sigma_K(\bSigma^{1/2})\leq \frac{\Delta/2}{Q^{-1}\left(\frac{\epsilon}{2}\right)}$, due to~\eqref{eq:A1implication}, yields the desired result.
\end{proof}

We now use Lemma~\ref{lem:fastconvergence} to show that is $P_m$ is sufficiently small, after a few more iterations of the (genie-aided) main algorithm the probability of missing $\CUBE$ will not be much greater than $\epsilon$.

\begin{corollary}
Consider the (genie-aided) main algorithm after $m$ iterations have been preformed. Let $\bA=\bA^{(m)}$ be the integer matrix at this point, define the random vector $\bV=\bV^{(m)}=\bA\bX$, and let $P_m\triangleq \Pr(\bV^{(m)}\notin\CUBE)$. Assume $Q^{-1}\left(\frac{\epsilon}{2}\right)\geq 6\sqrt{K}$ and $\beta<0.1$. If $P_m<\frac{1}{30}$, then
\begin{align}
P_{m+2}\leq K\cdot Q\left(0.99\sqrt{\frac{1-\beta}{1+\beta}}Q^{-1}\left(\frac{\epsilon}{2}\right)\right).
\end{align}
\label{cor:3iter}
\end{corollary}

\begin{proof}
It can be verified that if $P_m<1/30$ and $\beta<0.1$, then $\sqrt{\frac{1-\beta}{1+\beta}\frac{\alpha(P_m;K)}{1-P_m}}\geq\sqrt{\frac{0.9\cdot(1-\sqrt{1/10})}{1.1\cdot (1-1/30)}}\geq\frac{3}{4} $. Thus, applying Lemma~\ref{lem:fastconvergence}  and recalling that $Q^{-1}\left(\frac{\epsilon}{2}\right)\geq 6\sqrt{K}$, shows that
\begin{align}
P_{m+1}\leq K\cdot Q\left(\frac{9}{2}\sqrt{K}\right)\leq Q\left(\frac{9}{2}\right),
\end{align}
where we have used the monotonicity of $K\mapsto K\cdot Q\left(\frac{9}{2}\sqrt{K}\right)$ in the last inequality. Applying Lemma~\ref{lem:fastconvergence} again, we have that
\begin{align}
P_{m+2}
&\leq K\cdot Q\left(\sqrt{\frac{1-\beta}{1+\beta}}\cdot\sqrt{\frac{1-\sqrt{3Q(9/2)}}{1-Q(9/2)}}Q^{-1}\left(\frac{\epsilon}{2}\right)\right) \nonumber \\
&\leq Q\left(0.99\sqrt{\frac{1-\beta}{1+\beta}}Q^{-1}\left(\frac{\epsilon}{2}\right)\right),
\end{align}
as desired.
\end{proof}

Corollary~\ref{cor:3iter} leveraged Lemma~\ref{lem:fastconvergence} to show that once the probability of missing $\CUBE$ is not too large, it decreases very fast from iteration to iteration. Unfortunately, for large $P_m$, Lemma~\ref{lem:fastconvergence} does not imply that $P_{m+1}<P_m$, as $\frac{\alpha(P;K)}{1-P}$ is very large for $P$ close to $1$. To this end, we now develop another technique for upper bounding $P_{m+1}$ in terms of $P_m$, which is effective for large $P_m$ (but not for small $P_m$).

\begin{lemma}
	Consider the (genie-aided) main algorithm after $m$ iterations have been preformed. Let $\bA=\bA^{(m)}$ be the integer matrix at this point, define the random vector $\bV=\bV^{(m)}=\bA\bX$, and let $P_m\triangleq \Pr(\bV^{(m)}\notin\CUBE)$.  Then, for any $0<\gamma<1$ we have that
	\begin{align}
	P_{m+1}\leq G_{\gamma}\left((1-\delta_\gamma)(1-P_m)\right),
	\end{align}
	where
	\begin{align}
	G_{\gamma}(p)\triangleq 2Q\left(\frac{Q^{-1}\left(\frac{1-p}{2}\right)}{\gamma}\right).
	\end{align}
	and
	\begin{align}
	\delta_{\gamma}\triangleq K\frac{1+\beta}{1-\beta}\frac{1}{\left(\gamma Q^{-1}\left(\frac{\epsilon}{2}\right)\right)^2}
	\end{align}
	\label{lem:PmDynamics}
\end{lemma}

\begin{proof}
Let $0<\gamma<1$, and let $\tilde{\bA}=[\tilde{\ba}_1|\cdots|\tilde{\ba}_K]$ be the integer matrix found at step $2$ of the $m+1$th iteration of the main algorithm. Applying the union bound and Chebyshev/Markov inequality we have that
\begin{align*}
\Pr&\left(\tilde{\bA}\bV\notin\gamma\CUBE \ \bigg| \ \bV\in\CUBE\right) \\
&\leq \sum_{k=1}^K \Pr\left(|\tilde{\ba}_k^T\bV|\geq\gamma\frac{\Delta}{2} \ \bigg| \ \bV\in\CUBE\right)\nonumber\\
&=\sum_{k=1}^K \Pr\left(\left(\tilde{\ba}_k^T\bV\right)^2\geq\gamma^2\frac{\Delta^2}{4} \ \bigg| \ \bV\in\CUBE\right)\nonumber\\
&\leq \sum_{k=1}^K\frac{\tilde{\ba}_k^T\mathbb{E}\left[\bV\bV^{T}|\bV^{}\in\CUBE\right]\tilde{\ba}_k}{\gamma^2\frac{\Delta^2}{4}}\nonumber\\
&\leq K\frac{\max_{k\in[K]}\tilde{\ba}_k^T\mathbb{E}\left[\bV\bV^{T}|\bV\in\CUBE\right]\tilde{\ba}_k}{\gamma^2\frac{\Delta^2}{4}}\nonumber\\
&\leq K\frac{1+\beta}{1-\beta}\frac{1}{\left(\gamma Q^{-1}\left(\frac{\epsilon}{2}\right)\right)^2},\nonumber
\end{align*}
where we have used Lemma~\ref{lem:intqual} in the last inequality.
Thus
\begin{align}
\Pr(\tilde{\bA}\bV^{(m)}\in\gamma\CUBE|\bV^{(m)}\in\CUBE)\geq 1-\delta_{\gamma}.\label{eq:assumesubcube}
\end{align}
	
By definition of the main algorithm, $\bA^{(m+1)}=\tilde{\bA}\bA$, and consequently, $\bV^{(m+1)}=\tilde{\bA}\bV^{(m)}$. We can therefore write
{\small	\begin{align}
	\Pr\left(\bV^{(m+1)}\in\gamma\CUBE\right)&=\Pr(\tilde{\bA}\bV^{(m)}\in\gamma\CUBE)\nonumber\\
	&\geq\Pr\left(\tilde{\bA}\bV^{(m)}\in\gamma\mathrm{CUBE}\big|\bV^{(m)}\in\mathrm{CUBE}\right)\nonumber \\
	&\cdot \Pr\left(\bV^{(m)}\in\mathrm{CUBE}\right)\nonumber\\
	&\geq (1-\delta_{\gamma})(1-P_m),\label{eq:incubegammabound}
	\end{align}
}
	where in the last inequality we have used~\eqref{eq:assumesubcube}, and the definition of $P_m$. 
	
	Since $\gamma<1$, we have that $\gamma\CUBE\subset\CUBE$, and consequently $\Pr\left(\bV^{(m+1)}\in\CUBE\right)\geq \Pr\left(\bV^{(m+1)}\in\gamma\CUBE\right)$. This inequality holds for any random vector $\bV^{(m+1)}$. However, since $\bV^{(m+1)}$ is a Gaussian random vector, it is reasonable to expect that a stronger inequality holds, of the form $\Pr\left(\bV^{(m+1)}\in\CUBE\right)\geq F_{\gamma}\left(\Pr\left(\bV^{(m+1)}\in\gamma\CUBE\right)\right)$, for some function $t\mapsto F_{\gamma}(t)$ which satisfies $F_{\gamma}(t)> t$ if $0<\gamma<1$. Lata\l{}a and Oleszkiewicz~\cite{lo99} have shown that this is indeed the case, and found the best possible function $F_{\gamma}(t)$, independent of the covariance matrix.
	
	To be more precise, let $\bY\sim\m{N}\left(\mathbf{0},\bR \right)$ be a $K$-dimensional Gaussian vector. For $0<\gamma<1$, define the function
	\begin{align}
	F_{\gamma}(p)\triangleq\min_{\m{S},\bR \ s.t. \ \Pr(\bY\in\gamma\m{S})\geq p} \Pr(\bY\in\m{S}),\label{eq:Fgammadef}
	\end{align}
	where the minimization is over all convex symmetric sets $\m{S}\subset\RR^K$ and positive semi-definite matrices $\bR$. By~\cite[Corollary 1]{lo99}, we have that
	\begin{align}
	F_{\gamma}(p)=1-2Q\left(\frac{Q^{-1}\left(\frac{1-p}{2}\right)}{\gamma}\right),\label{eq:Fgammaformula}
	\end{align}
	and the minimum is attained, e.g., by taking $\bR=\bI$ and $\m{S}=\left\{\bx\in\RR^K \ : \  |x_1|\leq \frac{1}{\gamma}Q^{-1}\left(\frac{1-p}{2}\right) \right\}$. Thus, indeed the function $F_{\gamma}(p)$ is dimension independent. Now, since $\CUBE$ is a convex symmetric set, we must have that
	\begin{align}
	1-P_{m+1}&=\Pr\left(\bV^{(m+1)}\in\CUBE\right)\nonumber\\
	&\geq F_{\gamma}\left(\Pr\left(\bV^{(m+1)}\in\gamma\CUBE\right)\right)\nonumber\\
	&\geq F_{\gamma}\left((1-\delta_{\gamma})(1-P_m)\right),
	\end{align}
	where in the last inequality we have used~\eqref{eq:incubegammabound} and the monotonicity of $p\mapsto F_{\gamma}(p)$. Rearranging terms, establishes the claim. 
\end{proof}

We now turn to upper bound the function $G_{\gamma}(p)$ by a simpler function, that is more easy to handle. 

\begin{proposition}
For any $0<p<1/2$ and $0<\gamma<1$, we have that
\begin{align}
2Q\left(\frac{Q^{-1}(p)}{\gamma}\right)\leq (2p)^{\frac{1}{\gamma}}.
\end{align}
\label{prop:Gub}
\end{proposition}

\begin{proof}
The claim is equivalent to
\begin{align}
Q\left(\frac{1}{\gamma}Q^{-1}(p)\right)\leq 2^{\frac{1}{\gamma}-1}\left[Q(Q^{-1}(p))\right]^{\frac{1}{\gamma}}.\label{eq:Gequivalentform}
\end{align}
To establish the latter, let $t=t_p=Q^{-1}(p)>0$ and let $\mu$ denote the standard Gaussian measure on $\RR$. Recalling that the measure $\mu$ is log-concave, we have that
\begin{align}
Q(Q^{-1}(p))&=\mu\left([t,\infty]\right)\nonumber\\
&=\mu\left(\gamma\left[\frac{1}{\gamma}t,\infty\right]+(1-\gamma)[0,\infty]\right)\nonumber\\
&\geq \mu\left(\left[\frac{1}{\gamma}t,\infty\right]\right)^{\gamma}\mu\left([0,\infty]\right)^{1-\gamma}\label{eq:logconcG}\\
&=\left[Q\left(\frac{t}{\gamma}\right)\right]^{\gamma}\cdot \left[\frac{1}{2}\right]^{1-\gamma}\nonumber\\
&=\left[Q\left(\frac{Q^{-1}(p)}{\gamma}\right)\right]^{\gamma}\cdot 2^{\gamma-1}\nonumber,
\end{align}
where~\eqref{eq:logconcG} follows from the log-concavity of $\mu$, and the last equality follows by definition of $t$. Now, rearranging terms establishes~\eqref{eq:Gequivalentform}.
\end{proof}

The choice $\gamma=1/2$ simplifies analysis, and we will therefore restrict attention to this choice for the remainder.
Combining Lemma~\ref{lem:PmDynamics} and Proposition~\ref{prop:Gub}, we obtain the following corollary.
\begin{corollary}
	Consider the (genie-aided) main algorithm after $m$ iterations have been preformed. Let $\bA=\bA^{(m)}$ be the integer matrix at this point, define the random vector $\bV=\bV^{(m)}=\bA\bX$, and let $P_m\triangleq \Pr(\bV^{(m)}\notin\CUBE)$. We have that
	\begin{align}
	P_{m+1}\leq \Psi(P_m)\triangleq\left(1-(1-\delta)(1-P_m)\right)^2=\left(\delta+P_m(1-\delta)\right)^2,
	\end{align}
	where
	\begin{align}
	\delta\triangleq\frac{4K}{\left( Q^{-1}\left(\frac{\epsilon}{2}\right)\right)^2}\frac{1+\beta}{1-\beta}
	\end{align}
	\label{cor:dynamics}
	\end{corollary}
The following proposition gives simple properties of the function $\Psi(p)$.
\begin{proposition}
For $\delta<1/2$, the function $p\mapsto\Psi(p)=\left(\delta+p(1-\delta)\right)^2$ is convex, increasing, and has a unique fixed point 
\begin{align}
p^*=\left(\frac{\delta}{1-\delta}\right)^2
\end{align}
in the interval $[0,1)$. Furthermore, $\Psi(p)<p$ for all $p\in(p^*,1)$ and $\Psi(p)\leq p^*$ for all $p\in[0,p^*]$.
\end{proposition}

\begin{proof}
Convexity and monotonicity are trivial. Uniqueness of the fixed point $p^*$ follows from solving the quadratic equation $\Psi(p)-p=0$. The last two claims follow since $p-\Psi(p)<0$ for $p\in(p^*,1)$, and $\Psi(p)-p\leq 0$ in $(0,p^*)$ and therefore in this range $\Psi(p)<p<p^*$.
\end{proof}

\begin{proposition}
	Consider the (genie-aided) main algorithm after $m$ iterations have been preformed. Let $\bA=\bA^{(m)}$ be the integer matrix at this point, define the random vector $\bV=\bV^{(m)}=\bA\bX$, and let $P_m\triangleq \Pr(\bV^{(m)}\notin\CUBE)$. Assume that 
	\begin{align}
	\delta=\frac{4K}{\left( Q^{-1}\left(\frac{\epsilon}{2}\right)\right)^2}\frac{1+\beta}{1-\beta}<\frac{1}{2},\nonumber
	\end{align}
	and let $P=\Pr(\bX\notin\CUBE)$.
	We have that
	\begin{align}
	P_m<\left(\frac{\delta}{1-\delta}\right)^2+\left(\frac{3+P}{4}\right)^m.
	\end{align}
	\label{prop:dynamics}
\end{proposition}

\begin{proof}
By Corollary~\ref{cor:dynamics} we have that
{\small
	\begin{align}
P_{m+1}-\left(\frac{\delta}{1-\delta}\right)^2 &\leq \left(\delta+P_m(1-\delta)\right)^2-\left(\frac{\delta}{1-\delta}\right)^2\nonumber\\
&=\left(1-(1-\delta)^2(1-P_m)\right)\left(P_m-\left(\frac{\delta}{1-\delta}\right)^2\right).
\end{align}
}
We may assume without loss of generality that $P_m>\left(\frac{\delta}{1-\delta}\right)^2$, as otherwise we have that $P_m\leq \left(\frac{\delta}{1-\delta}\right)^2$ due to Proposition~\ref{prop:dynamics},
and our claim follows trivially. Moreover, since $\delta<1/2$, we have that $1-(1-\delta)^2(1-P_m)\leq \frac{3+P_m}{4}$. 
Using Proposition~\ref{prop:dynamics} again, we have that $P_m\leq P$ for all $m$, and consequently
\begin{align}
P_{m+1}-\left(\frac{\delta}{1-\delta}\right)^2&\leq \left(\frac{3+P}{4}\right)\left(P_m-\left(\frac{\delta}{1-\delta}\right)^2\right)\nonumber\\
&\leq \left(\frac{3+P}{4}\right)^{m+1}\left(P-\left(\frac{\delta}{1-\delta}\right)^2\right)\nonumber\\
&\leq \left(\frac{3+P}{4}\right)^{m+1}.
\end{align}
\end{proof}

%

We are now ready to prove Theorem~\ref{thm:genie}.

\begin{proof}[Proof of Theorem 1]
Assume $Q^{-1}\left(\frac{\epsilon}{2}\right)\geq 6\sqrt{K}$ and $\beta<0.1$, such that
\begin{align}
\delta=\frac{4K}{\left(Q{-1}\left(\frac{\epsilon}{2}\right)\right)^2}\frac{1+\beta}{1-\beta}\leq \frac{1\cdot 1.1}{9\cdot 0.9}\leq \frac{1}{7}.
\end{align}
By Proposition~\ref{prop:dynamics}, after $m=\frac{\log{180}}{\log\left(\frac{4}{3+P}\right)}$ iterations we have that 
\begin{align}
P_m\leq \left(\frac{\delta}{1-\delta}\right)^2+\frac{1}{180}\leq\left(\frac{1/7}{6/7}\right)^2+\frac{1}{180}=\frac{1}{30}.
\end{align}
Now, applying Corollary~\ref{cor:3iter}, we see that 
\begin{align}
P_{m+2}\leq K\cdot Q\left(0.99\sqrt{\frac{1-\beta}{1+\beta}}\cdot Q^{-1}\left(\frac{\epsilon}{2}\right) \right),
\end{align}
as desired.
\end{proof}

\vspace{6mm}

\section{Analysis of Truncated Covariance Estimation Algorithm}
\label{sec:TruncEst}

In this section we analyze the proposed algorithm for estimating the truncated covariance matrix. 
We emphasize that our method is one of possibly many methods to perform this task, each of whom can be plugged into the pipeline outlined before instead of this particular algorithm. Hence, the results in this section are completely standalone with respect to the rest of this paper.

\paragraph*{Outline of the argument} The main ideas of the analysis are as follows:
\begin{enumerate}
	\item We first make the geometric observation that if the error probability of the optimal (MAP) estimator is sufficiently small, then with very high probability the unfolded sample $\bX$ falls onto a relatively small ellipsoid, 
	\[
	\m{P}_r = \left\{ \bx\,:\,\bx^T \bSigma^{-1}\bx \le r^2 \right\}
	\] for some appropriate $r>0$. Now, given that two points $\bX_i,\bX_j\in\m{P}_r$, where $\bX_i\in\CUBE$ but $\bX_j\notin\CUBE$ (so that $\bX_j^* \ne \bX_j$), we have a lower bound $\norm{\bX_i^*-\bX_j^*}>d$, for some appropriate $d$. Hence, provided that the entire sample $\bX_1,\ldots,\bX_n \in \m{P}_r$ (this happens with high probability when $n$ is sufficiently small), any observed point $\bX_i^*$ that is $d$-connected to $\mathbf{0}$ must satisfy that $\bX_i\in\CUBE$, and therefor $\bX_i^*=\bX_i$ - that is, the points that the algorithm use for estimation are indeed all samples from the truncated Gaussian distribution. By $d$-connected, we mean that there is a sequence $\bX_{i_0}^*=\mathbf{0} \to \bX_{i_1}^* \to \ldots \to \bX{i_m}=\bX_i^*$ where the $\bX_{i_j}^*$-s are all sample points (besides $j=0$) and the distance between each consecutive pair is at most $d$.
	\item Now we need to show that our algorithm indeed identifies \emph{most} (that is, all but $o(n)$) of the sample points such that $\bX_i\in\CUBE$. We do this using a covering argument, and this turns out to be the most technically challenging part of the analysis. We also suspect that the bounds our proof technique yields are quite sub-optimal. 
	\item Letting $\m{T}$ be the set of unfolded points taken by the algorithm, and $\m{S}_0$ be the (unknown) set of sample points $i$ with $\bX_i\in\CUBE$, we conclude that 
	\[
	\frac{1}{\abs{\m{T}}} \sum_{i\in\m{T}} \bX_i^*(\bX_i^*)^T \approx \frac{1}{\abs{\m{S}_0}}\sum_{i\in \m{S}_0} \bX_i^*(\bX_i^*)^T \,.
	\]
	But now, the expression on the right is clearly a consistent estimator for the desired covariance matrix $\E{\bX \bX^T | \bX\in\CUBE}$. 
\end{enumerate}

For the rest of this section, we assume the matrix $\bSigma$ satisfies assumptions $A1$, $A2$ and $A3$ with parameters $\epsilon>0$, $\tau_{\text{min}}>0$, and $0<P<1$, respectively.

\subsection{Geometric preliminaries: no false positives with high probability}

First, we use Proposition~\ref{prop:contained_ball} for lower bounding the packing radius of the lattice $\Lambda\left(\Delta\bSigma^{-1/2}\right)$.

\begin{proposition}
	Let 
	\begin{align}
	r_0=r_0\left(\Lambda\left(\Delta\bSigma^{-1/2}\right)\right)\triangleq \frac{\min_{\bb\in\ZZ^{K}\setminus\{\mathbf{0}\}}\Delta\|\bSigma^{-1/2}\bb\|}{2}.
\end{align}
	 be the packing radius of the lattice $\Lambda\left(\Delta\bSigma^{-1/2}\right)$. Then,
	\begin{align}
	r_0\geq Q^{-1}\left(\frac{\epsilon}{2}\right).\label{eq:r0}
	\end{align}
	\label{prop:rpack}
\end{proposition}

\begin{proof}
	Recall that in Subsection~\ref{subsec:IFD} we have shown that
	\begin{align}
	\min_{\bA\in\SLk}\Pr\left(\bA\bX\notin\CUBE\right)\geq \Pr(g_{\text{MAP}}(\bX^*)\neq\bX).
	\end{align}
	Thus, assumption $A1$ combined with~\eqref{eq:PeMAP}, implies that
	\begin{align}
	\Pr\left(\bZ\notin\m{V}(\Delta\bSigma^{-1/2})\right)\leq\epsilon.\label{eq:PeVor}
	\end{align}
	Noting that $\m{V}$ is a convex, symmetric body, and that $r_0$ is the radius of the largest, centered ball contained in $\m{V}$, the required result now follows immediately from Proposition \ref{prop:contained_ball}.
\end{proof}

Introduce the norm
\begin{align}
\|\bx\|_{\bSigma}\triangleq\sqrt{\bx^T\bSigma^{-1}\bx},
\end{align}
and note that $\norm{\bx} \ge \tau_{\text{min}}\norm{\bx}_\bSigma$ by Assumption $A2$. Denote the corresponding ball by
\begin{align}
\m{P}_r = \left\{ \bx\,:\,\norm{\bx}_\bSigma^2 \le r^2 \right\} \,.
\end{align}

Recall from Proposition \ref{prop:rpack} that the largest ellipsoid $\m{P}_{r_0}$ contained in $\m{R}_{\text{MAP}}$ has radius
\[
r_0 = \frac{1}{2}\Delta \min_{\bb\in\ZZ^k,b\ne 0} \norm{\bb}_\bSigma \geq Q^{-1}\left(\frac{\epsilon}{2}\right) \,.
\]

\begin{lemma}
	Let $\eta\in(0,1)$, and suppose that $\bx,\by\in\CUBE$ and $0\ne\bb\in\ZZ^k$ are such that
	$\bx,\by+\Delta \bb \in \m{P}_{(1-\eta)r_0}$. Then
	\begin{align*}
	\norm{\bx-\by} \ge 2\eta \cdot \tau_{\text{min}} \cdot  Q^{-1}\left(\frac{\epsilon}{2}\right) \,.
	\end{align*}
\end{lemma}
\begin{proof}
	We have that
	\begin{align*}
	\frac{1}{\tau_{\text{min}}}\cdot \norm{\bx-\by}
	&\ge \norm{\bx-\by}_\bSigma \\
	&= \norm{\bx-(\by+\Delta \bb) +\Delta \bb}_\bSigma \\
	&\ge   \norm{\Delta \bb}_\bSigma - \norm{\bx}_\bSigma -\norm{(\by+\Delta \bb)}_\bSigma \\
	&\ge 2\eta \cdot r_0 \\
	&\ge 2\eta \cdot  Q^{-1}\left(\frac{\epsilon}{2}\right)\,,
	\end{align*}
	where we used the triangle inequality, $\norm{\bb}_\bSigma \ge 2r_0/\Delta$, and the assumption that  $\bx,\by+\Delta \bb \in \m{P}_{(1-\eta)r_0}$.
\end{proof}

We state the following Chernoff-type bound for $\chi^2$-distributed random variables, which will be used throughout. 
\begin{proposition}[Proposition 13.1.3 in~\cite{ramibook}]
	For $\bZ\sim\m{N}(\mathbf{0},\bI_K)$ and $r\ge \sqrt{K}$ it holds that
	\begin{align}
	\Pr(\norm{\bZ} \ge r) \le e^{-\frac{K}{2}f\left(\frac{r}{\sqrt{K}}\right)}, 
	\end{align}
	where
	\begin{align}
	f(\alpha)\triangleq\alpha^2-1-\ln{\alpha^2}.\label{eq:fdef}
	\end{align}
	\label{prop:poltyrev}
\end{proposition}

The following is an immediate corollary.

\begin{lemma}
	\label{lem:no-false-positives}
	Let $\kappa_{\epsilon}\triangleq \frac{Q^{-1}\left(\frac{\epsilon}{2}\right)}{\sqrt{K}}$, and assume $\kappa_{\epsilon}>1$. Suppose that the proposed truncated covariance estimation algorithm is run with a distance parameter satisfying
	\begin{equation}
	d \le 2\sqrt{K} \cdot \tau_{\text{min}} \cdot \kappa_{\epsilon}\cdot \eta \,,
	\end{equation}
	for some $0 < \eta < 1-\frac{1}{\kappa_{\epsilon}}$. 
	Let $\m{E}_{\text{false-positive}}$ be the event that the algorithm classifies incorrectly a \emph{false positive}, namely, that there exists a sample point $\bX_i^* \in \CUBE$ that is $d$-connected to $\mathbf{0}$ but $\bX_i\notin\CUBE$. Then
	\begin{align}
	\Pr\left(\m{E}_{\text{false-positive}}\right) \le ne^{-\frac{K}{2}f((1-\eta)\kappa_{\epsilon})} \,,
	\end{align}
	where $f(\cdot)$ is the rate function from Proposition \ref{prop:poltyrev}.
\end{lemma}
\begin{proof}
	By the previous lemma, 
	\begin{align}
	\Pr(\m{E}_{\text{false-positive}}) 
	&\le \Pr \left(\exists \bX_i \not\in \m{P}_{(1-\eta)r_0} \right)\nonumber \\
	&\le n \Pr \left( \norm{\bX}_\bSigma> (1-\eta)r_0 \right)\nonumber \\
	&= n \Pr \left( \norm{\bZ} > (1-\eta) r_0 \right)\nonumber \\
	&\le n \Pr \left(  \norm{\bZ} > (1-\eta) Q^{-1}\left(\frac{\epsilon}{2}\right) \right)\label{eq:r0ineq} \\
	&= n \Pr \left(  \norm{\bZ} > (1-\eta) \kappa_{\epsilon}\sqrt{K} \right)\,,\nonumber
	\end{align}
	where $\bZ\sim \m{N}(\mathbf{0},\bI_K)$, and we have used Proposition~\ref{prop:rpack} in~\eqref{eq:r0ineq}. Now, by assumption $(1-\eta)\kappa_{\epsilon} > 1$, and we may use Proposition \ref{prop:poltyrev}.
\end{proof}

\subsection{The algorithm identifies most of the unfolded points}

The idea is the following. We find an ellipsoid $\m{P}_{\overline{r}}$ that contains, with high probability, all but $O\left(\sqrt{n}\right)$ of the points $\bX_i$. Moreover, provided that $n$ is sufficiently large, if we cover $\m{P}_{\overline{r}}\cap\CUBE$ by balls of radii $d/4$, then with high probability \emph{every} such ball contains a point $\bX_i^*$. This, then, implies that all the points contained in $\m{P}_{\overline{r}}\cap\CUBE$ must be $d$-connected to $\mathbf{0}$, and so we miss at most the $O(\sqrt{n})$ points that lie outside $\m{P}_{\overline{r}}$.

\begin{lemma}[Size of $(d/4)$-cover]
	\label{lem:cover-size}
	There exists a $d/4$-cover of $\m{P}_{\overline{r}} \cap \CUBE$ of size at most 
	\begin{equation}
	\abs{\m{C}} \le \left(\frac{4\Delta}{d} + 1\right)^K K^{K/2} \,.
	\end{equation} 
\end{lemma}
\begin{proof}
	The following is a standard Gilbert-Varshamov type argument. Let $\m{C} \subset \m{P}_{\overline{r}}\cap \CUBE$ be a maximal $d/4$-separated set, that is, for any $\bx,\by\in\m{C}$ with $\bx\ne\by$ we have $\norm{\bx-\by}>d/4$, and moreover $\m{C}$ cannot be extended into a strictly larger $d/4$-separated set. In particular, $\m{C}$ is a $d/4$-cover of $\m{P}_{\overline{r}} \cap \CUBE$, since if $\exists \bw\in\m{P}_{\overline{r}} \cap \CUBE$ that is not $d/4$-covered, we can add it to $\m{C}$, constradicting its maximality. Since the balls $\left\{ \m{B}(\bx,d/8) \right\}_{\bx\in\m{C}}$ are all disjoint and contained in
	\begin{align*}
		\m{P}_{\overline{r}} \cap \CUBE + (d/8)B 
		&\subset \CUBE + (d/8)B \\
		&= (\Delta/2)B_{l^\infty} + (d/8)B \\
		&\subset (\Delta/2+d/8)B_{l^\infty} 
	\end{align*}
	(here $B$ is the $l^2$ unit ball, and $B_{l^{\infty}}$ is the $l^{\infty}$ unit ball, and of course $B \subset B_{l^\infty}$), we have that 
	\begin{align*}
	\abs{\m{C}} 
	&\le \frac{\Vol(\m{P}_{\overline{r}} \cap \CUBE + (d/8)B)}{\Vol((d/8)B)} \\
	&\le \left(\frac{(\Delta/2+d/8}{d/8}\right)^K \frac{\Vol(B_{l^\infty})}{\Vol(B)}\\
	&\le  \left(\frac{4\Delta}{d} + 1\right)^K K^{K/2} \,,
	\end{align*}
	where we used the crude estimate $(1/\sqrt{K})B_{l^\infty} \subset B$ so $V_K=\Vol(B) \ge K^{-K/2}\Vol(B_{l^\infty})$.
\end{proof}

\begin{lemma}[Measure of a small ball]
	\label{lem:small-ball}
	Let $\bx\in\m{P}_{\overline{r}}\cap\CUBE$, and suppose that $\overline{r} \ge \sqrt{K}$ and $d \le \Delta\sqrt{K}/2$ (note that larger values of $d$ are meaningless, since every point in $\CUBE$ is within distance $\Delta\sqrt{K}/2$ from $\mathbf{0}$). Then
	\begin{equation}
		\begin{split}
				\Pr&(\bX\in \m{B}(\bx,d/4)\cap\CUBE) \\
			&\ge \left(\frac{d}{\sqrt{2\pi K} \left(\frac{4\Delta}{2\tau_{\text{min}}} + 4\right)}\right)^K \abs{\bSigma}^{-1/2} e^{-\overline{r}^2/2} \,.
		\end{split}
	\end{equation}
\end{lemma}
\begin{proof}
	First, observe that if $0<b\le \Delta$ and $\by\in\CUBE$, then 
	\[
	\Vol(\m{B}(\by,b)\cap\CUBE) \ge 2^{-K}\Vol(\m{B}(\by,b))\,,
	\]
	the worst case being when $\by$ is one of the extremal points (vertices) of the cube. Let us now find some other ball $\m{B}(\by,b) \subset \m{B}(\bx,d/4)\cap\m{P}_{\overline{r}}$ with $\by\in\CUBE$. We try a center of the form $\by=(1-\delta)\bx$, for some $\delta\in(0,1)$ to be determined. 
	In order to have $\m{B}(\by,b)\subset \m{B}(\bx,d/4)$, we need that for all $\norm{\bu}\le 1$, it holds that $\norm{\bx-(\by+b\bu)} \le d/4$. Optimizing for the worst $\bu$ and putting $\bx-\by=\delta\bx$, we need $\delta\norm{\bx}+b \le d/4$. Since $\bx\in\CUBE$, we have $\norm{\bx}\le \Delta\sqrt{K}/2$, so a sufficient condition is 
	\begin{equation*}
	b + \delta\cdot \Delta\sqrt{K}/2 \le d/4 \,.
	\end{equation*}
	As for the condition $\m{B}(\by,b)\subset \m{P}_{\overline{r}}$, we need $\norm{\by+b\bu}_\bSigma \le \overline{r}$, so that using $\norm{\bu}_\bSigma \le 1/\tau_{\text{min}}$, and $\norm{\by}_\bSigma = (1-\delta)\norm{\bx}_\bSigma \le (1-\delta)\overline{r}$, we have for a sufficient condition
	\begin{equation*}
	b \le \tau_{\text{min}}\cdot\delta\cdot\overline{r} \,.
	\end{equation*}
	Then, in order to satisfy both conditions, we may take 
	\begin{equation}
	b = \tau_{\text{min}}\cdot\delta\cdot\overline{r} \,,
	\end{equation}
	and
	\begin{equation}
	\delta = \frac{d/4}{\Delta\sqrt{K}/2 + \tau_{\text{min}}\overline{r}} \,,
	\end{equation}
	with $\delta<1$ since $d\le \Delta\sqrt{K}/2$. Now, we note that the density of $\bX$ is lower bounded on $\m{B}(\by,b)$ by $(2\pi)^{-K/2}\abs{\bSigma}^{-1/2}e^{-\overline{r}^2/2}$, so that
	\begin{equation*}
	\begin{split}
	\Pr&(\bX\in \m{B}(\bx,d/4)\cap\CUBE) \\
	&\ge (2\pi)^{-K/2}\abs{\bSigma}^{-1/2}e^{-\overline{r}^2/2} \cdot \Vol(\m{B}(\by,b)\cap\CUBE)\\
	&\ge (2\pi)^{-K/2}\abs{\bSigma}^{-1/2}e^{-\overline{r}^2/2} \cdot 2^{-K} \Vol(\m{B}(\by,b)) \\
	&\ge (2\pi)^{-K/2}\abs{\bSigma}^{-1/2}e^{-\overline{r}^2/2} \cdot 2^{-K} \cdot b^K \cdot 2^K K^{-K/2} \,,
	\end{split}
	\end{equation*}
	where we used the estimate $\Vol(B) \ge K^{-K/2}\Vol(B_{l^\infty}) =K^{-K/2}2^K$ for the $l^2$ unit ball. Note that 
	\begin{align*}
	b = \frac{\tau_{\text{min}}\cdot \overline{r}\cdot d/4}{\Delta\sqrt{K}/2 + \tau_{\text{min}}\overline{r}} = \frac{d}{\frac{4\Delta}{2\tau_{\text{min}}}\cdot \sqrt{K}/\overline{r} + 4} \ge \frac{d}{\frac{2\Delta}{\tau_{\text{min}}} + 4}\,,
	\end{align*}
	where we used the assumption $\overline{r}\ge\sqrt{K}$. Putting this estimate into the expression above gives the required result.
\end{proof}

From here on, we take 
\begin{equation}
\label{eq:r-cover}
\overline{r}^2 = K + \log(n) + \sqrt{2K\log(n)} \,.
\end{equation}

\begin{lemma}
	\label{lem:cover-existence}
	Fix a $(d/4)$-cover, $\m{C}$ of $\m{P}_{\overline{r}}\cap\CUBE$ such that $\abs{\m{C}}$ is minimal. Let $\m{E}_{\text{miss-cover}}$ be the event that 
	there exists some $\bx\in\m{C}$ such that $\m{B}(\bx,d/4)$ doesn't contain any sample point $\bX_i^*$. Then
	{\footnotesize
	\begin{equation}
		\begin{split}
			\Pr \left(\m{E}_{\text{miss-cover}}\right) 
			\le 
			\exp \Biggl[& -\left(\frac{d}{\sqrt{2\pi K} \left(\frac{2\Delta}{\tau_{\text{min}}} + 4\right)}\right)^K \\
			&\cdot \abs{\bSigma}^{-1/2} e^{\frac{1}{2}\log(n) - \sqrt{\frac{1}{2}K\log(n)} - \frac{1}{2}K} \\
			&+K\log\left(\sqrt{K}+\frac{4\sqrt{K}\Delta}{d}\right)  \Biggr] \,.
				\end{split}
	\end{equation}
}
\end{lemma}
\begin{remark}[Curse of dimensionality]
	Holding all the parameters beside $n$ fixed, we see that as $n\to \infty$ the bound on the right hand side tends to $0$.  However, in order for the bound to become meaningful, the sample complexity $n$ must scale (super)-exponentially with $K$.
\end{remark}
\begin{proof}
	Let $q$ be the lower bound from  Lemma \ref{lem:small-ball}. The probability to miss any one ball in the cover is upper bounded by $(1-q)^n \le e^{-qn}$, so that using the union bound,
	\[
	\Pr \left(\m{E}_{\text{miss-cover}}\right) \le e^{-qn+\log\abs{\m{C}}} \,.
	\]  
	Putting the exact expressions for $q$, $\overline{r}$, and the upper bound on $\abs{\m{C}}$ from Lemma \ref{lem:cover-size}, we get the claimed bound.
\end{proof}

The following easy Lemma shows that if the event $\m{E}_{\text{miss-cover}}$ did not occur, we must have that any $\bx\in\m{P}_{\overline{r}}\cap\CUBE$ is $d$-connected to the origin.
\begin{lemma}
	\label{lem:cover_implies_connectedness}
	Let $\m{S}\subset \RR^K$ be some connected set, $\m{C}\subset \m{S}$ be a $d/4$-cover and $\m{A}$ be some discrete set of points so that every ball $\m{B}(\bx,d/4)$, $\bx\in \m{C}$ contains at least one point of $\m{A}$. Then every pair of points $\ba,\bb\in\m{A}$ is $d$-connected along $\m{A}$.
\end{lemma}
\begin{proof}
	Fix some continuous path $\psi : [0,1]\to\m{S}$ such that $\psi(0)=\ba$ and $\psi(1)=\bb$. Let 
	\[
	\hat{\psi}(t) = \arg\min_{\hat{\ba}\in\m{A}} \norm{\hat{\ba}-\psi(t)} \,,
	\]
	where ties are broken so that $\hat{\psi}$ is right-continuous, and denote $t_0=0$ and $t_1<t_2<\ldots$ the times where $\hat{\psi}(t)$ changes.\footnote{That is, $\hat{\psi}$ is piece-wise constant on the half-open intervals $[t_i,t_{i+1})$.} It would clearly suffice to show that for every $i=0,1,\ldots$, $\norm{\hat{\psi}(t_i)-\hat{\psi}(t_{i+1})} \le d$. Indeed, there is a point $\bc_i\in\m{C}$ in the cover such that $\norm{\bc_i-\psi(t_i)}\le d/4$, and since $\m{B}(\bc_i,d/4)$ contains a point from $\m{A}$, we must have that $\norm{\hat{\psi}(t_i)-\psi(t_{i})}\le d/2$. Observe also that at a switching time $t_{i+1}$,
	\[
	\norm{\hat{\psi}(t_i)-\psi(t_{i+1})} = \norm{\hat{\psi}(t_{i+1})-\psi(t_{i+1})} \,.
	\]
	Thus, by the triangle inequality,
	\[
	\norm{\hat{\psi}(t_i)-\hat{\psi}(t_{i+1})} \le 2\cdot \norm{\hat{\psi}(t_{i+1})-\psi(t_{i+1})} \le d \,.
	\]
\end{proof}

\begin{lemma}
	\label{lem:bound-escapees}
	Let $\m{E}_{\text{many-escapees}}$ be the event that there are more than $2\sqrt{n}$ points $\bX_i$ such that $\bX_i \notin \m{P}_{\overline{r}}$. Then
	\begin{equation}
	\Pr \left(\m{E}_{\text{many-escapees}}  \right) \le e^{-\frac{3}{14}\sqrt{n}} \,.
	\end{equation}  
\end{lemma}
\begin{proof}
	It is convenient to use the following tail bound for $\chi^2$ random variables, due to (\cite{laurent2000adaptive}, Lemma 1):  for $\bZ\sim\m{N}(\mathbf{0},\bI_K)$ and $x>0$, we have
	\begin{equation}
	\label{eq:tail-laurent}
	\Pr\left(\norm{\bZ}^2 \ge K + 2\sqrt{Kx} + 2x \right) \le e^{-x} \,.
	\end{equation}
	Putting $x=\frac{1}{2}\log(n)$, we get
	\begin{equation}
	\Pr\left(\bX \notin \m{P}_{\overline{r}} \right) \le n^{-1/2} \triangleq q \,.
	\end{equation} 
	Let $M$ be the number of points not in $\m{P}_{\overline{r}}$. Then $M \sim \Bin(n,q)$, so that by Bernstein's inequality (see, e.g, Theorem 2.8.4 in \cite{vershynin2017high}),
	\[
	\Pr(M > qn + t) \le \exp \left( \frac{\frac{1}{2}t^2}{2nq(1-q)+\frac{1}{3}t} \right)\,.
	\]
	Putting $t=n^{1/2}$, $qn=n^{1/2}$ we get 
	\[
	\Pr(M>2\sqrt{n}) \le \exp\left(-\frac{\sqrt{n}/2}{2(1-1/\sqrt{n})+1/3}\right) \le e^{-\frac{3}{14}\sqrt{n}} \,.
	\]
\end{proof}

\subsection{Estimating the covariance}

\paragraph*{Notation} We denote for brevity $\bSigma_{\text{truc}}=\E{\bX\bX^T|\bX\in\CUBE}$, and let $P\in(0,1)$ be a number such that $\Pr(\bX\notin\CUBE) \le P$. Let $\m{S}_0$ be the set of sample points $i$ such that $\bX_i=\bX_i^*\in\CUBE$. Also, let $\m{T}$ be the set of points $i$ such that $\bX_i^*$ is $d$-connected to $\mathbf{0}$. For a matrix $\bG$, we denote the operator norm as $\|\bG\|$.

Recall that our algorithm returns the following estimate for $\bSigma_{\text{truc}}$:
\[
\bSigmaV = \frac{1}{\abs{\m{T}}}\sum_{i\in\m{T}}\bX_i^{*}(\bX_i^*)^T \,.
\]

We start by analyzing the oracle estimator 
\[
\bSigmaV_{\text{oracle}} = \frac{1}{\abs{\m{S}_0}}\sum_{i\in \m{S}_0}\bX_i^{*}(\bX_i^*)^T \,.
\]
We'll show that with high probability $\norm{\bSigmaV_{\text{oracle}}-\bSigma_{\text{truc}}}$ is small. The results we proved in the previous subsections imply that $\norm{\bSigmaV_{\text{oracle}}-\bSigmaV}$ is small, and therefor $\bSigmaV$ is also a good estimator for $\bSigma_{\text{truc}}$.

\begin{lemma}[Sufficiently many points in $\CUBE$]
	\label{lem:suff_many_points}
	We have 
	\begin{equation}
	\Pr \left( \abs{\m{S}_0} < (1-P) n/2 \right) \le e^{-\frac{1}{2}(1-P)^2n} \,.
	\end{equation}
\end{lemma}
\begin{proof}
	This is an immediate application of Hoeffding's inequality.
\end{proof}

Note that conditioned on $\m{S}_0$, the points $\left\{\bX_i=\bX_i^* \right\}_{i\in\m{S}_0}$ are an i.i.d sample from the truncated Gaussian distribution. It will be convenient for our purposes to instead consider the transformed points $\bY_i = \bSigma_{\text{truc}}^{-1/2}\bX_i$, so that 
$\Cov(\bY_i)=\bI_K$, and also $\norm{\bY}^2 \le \frac{K\Delta^2}{4}\norm{\bSigma_{\text{truc}}^{-1}} \triangleq R^2$.   

\begin{lemma}
	\label{lem:vershynin-estimation}
	Fix an error parameter $\beta\in(0,1)$, and assume that $(1-P) n/2\ge CR^2\log(K)\beta^{-2}$ where $C>0$ is a universal constant. Let $\m{E}_{\text{sample-est}}$ be the event that 
	\begin{equation}
	\norm{ \frac{1}{\abs{\m{S}_0}}\sum_{i\in \m{S}_0}\bY_i^{*}(\bY_i^*)^T - \bI_K} > \beta \,.
	\end{equation}
	Then
	\begin{equation}
	\Pr \left( \m{E}_{\text{sample-est}} \right) \le e^{-\beta^2 \frac{(1-P) n}{2CR^2}} + e^{-\frac{1}{2}(1-P)^2n} \,.
	\end{equation}
\end{lemma}
\begin{proof}
	This follows directly from corollary 5.52 in \cite{vershynin2010introduction}, where we also need to intersect with the event $\abs{\m{S}_0}\ge (1-P) n/2$; using Lemma~\ref{lem:suff_many_points}, this gives us the second term in the bound.
\end{proof}

\begin{lemma}
	\label{lem:est-set-diff}
	Suppose that $\m{T}\subset \m{S}_0$. Then
	\begin{equation}
	\norm{ \frac{1}{\abs{\m{S}_0}}\sum_{i\in \m{S}_0}\bY_i^{*}(\bY_i^*)^T - \frac{1}{\abs{\m{T}}}\sum_{i\in \m{T}}\bY_i^{*}(\bY_i^*)^T } \le
	\frac{\abs{\m{S}_0}-\abs{\m{T}}}{\abs{\m{S}_0}}\cdot 2R^2 \,.
	\end{equation}
\end{lemma}
\begin{proof}
	Simply decompose 
	\[
	\sum_{i\in \m{S}_0}\bY_i^{*}(\bY_i^*)^T = \sum_{i\in \m{S}_0\setminus\m{T}}\bY_i^{*}(\bY_i^*)^T + \sum_{i\in \m{T}}\bY_i^{*}(\bY_i^*)^T
	\]
	and use
	\[
	\norm{\bY_i^{*}(\bY_i^*)^T} \le \norm{\bY_i^{*}}^2 \le R^2 \,.
	\]
\end{proof}

We are ready to  prove the main result of this section.

\begin{proof}[Proof of Theorem~\ref{thm:cov_est_ugly_bnd}]
	By Lemma \ref{lem:no-false-positives}, under the complement event $\overline{\m{E}_{\text{false-positive}}}$ we have $\m{T}\subset \m{S}_0$. By Lemmas \ref{lem:cover-existence}, \ref{lem:cover_implies_connectedness} and \ref{lem:bound-escapees}, under $\overline{ \m{E}_{\text{miss-cover}} \cup \m{E}_{\text{many-escapees}} }$ we have that $\abs{\m{S}_0\setminus \m{T}} \le 2\sqrt{n}$. We may now use Lemmas \ref{lem:vershynin-estimation} and \ref{lem:est-set-diff} to control the estimation error, where, using Lemma \ref{lemma:truncvar}, we may bound $R^2$ by
	\begin{align*}
		R^2 
		&= \frac{K\Delta^2}{4}\cdot \norm{\bSigma_{\text{truc}}^{-1}} \\
		&\le \frac{K\Delta^2}{4} \cdot \frac{1-P}{\alpha(P;K)}\norm{\bSigma^{-1}} \\
		&= \frac{K}{4} \cdot \left(\frac{\Delta}{\tau_{\text{min}}}\right)^2 \cdot \frac{1-P}{\alpha(P;K)}\,,
	\end{align*}
	and bound the estimation error $\overline{\beta}$ as 
	\begin{align*}
		\overline{\beta} 
		&= \beta + \frac{\abs{\m{S}_0}-\abs{\m{T}}}{\abs{\m{S}_0}}\cdot 2R^2 \\
		&\le \beta + \frac{8}{(1-P)\sqrt{n}}R^2 \\
		&\le \beta +  \frac{2K \cdot \left(\frac{\Delta}{\tau_{\text{min}}}\right)^2 }{\alpha(P;K)\cdot \sqrt{n}} \,,
	\end{align*}
	where we also used that under $\overline{\m{E}_{\text{sample-err}}}$ we have $\abs{\m{S}_0} \ge \frac{1}{2}(1-P)n$. By the union bound, the failure probability can be bounded by
	\begin{align*}
\Pr&(\m{E}_{\text{false-positive}}) + \Pr(\m{E}_{\text{many-escapees}}) \\
 &+ \Pr \left( \m{E}_{\text{sample-est}} \right) + \Pr(\m{E}_{\text{miss-cover}}) \,.
	\end{align*}
	Plugging the bound on $R^2$ and our choice of $d$ into the formulae we had in the previous Lemmas, we get the claimed bound.
\end{proof}

\section{Complete analysis of the modified algorithm}
\label{sec:combined}

We now wish to give probability guarantees on the run of the entire recovery algorithm, throughout \emph{all} the iterations. We cannot, however, use the bound in Theorem~\ref{thm:cov_est_ugly_bnd} as is with $\tau_{\text{min}}\leq\lambda_{\text{min}}(\bSigma^{1/2})$): in every iteration, we are estimating a different covariance matrix, and do not fully understand how $\lambda_{\text{min}}(\bSigma^{1/2})$ evolves; we cannot guarantee, for example, that it is increasing. Consequently, assumption $A2$ is not very useful. On the other hand, while the probability of missing $\CUBE$ also changes from iteration to iteration, our analysis in Section~\ref{sec:genieaided} showed that it can only decrease (unless it is already small to begin with). Thus, recalling that the lattice generated by $\bA\bSigma^{1/2}$ is identical to the lattice generated by $\bSigma^{1/2}$, for $\bA\in\SLk$, we have that assumptions $A1$, $A3$ and $A4$ remain valid throughout all iterations of the algorithm. We therefore begin this section with a few technical claims that show how $\epsilon$, $P$ and $\rho_{\text{pack}}$ control the smallest eigenvalue of $\bSigma$.

\subsection{From Assumptions $A1$, $A3$ and $A4$ to assumption $A2$}

We first show that under assumption $A1$, it is possible to upper bound the ratio $\frac{|\bSigma|^{1/2K}}{\Delta}$ in terms of $\epsilon$.
\begin{proposition}
	\label{prop:upper_bound_no_psi}
	If $\bSigma$ satisfies assumption $A1$, we have that
	\begin{align}
	\frac{|\bSigma|^{1/2K}}{\Delta}\leq \frac{1}{V_K^{1/K}\cdot Q^{-1}\left(\frac{\epsilon}{2}\right) }.
	\end{align}
	\begin{proof}
		Clearly $\m{B}(\mathbf{0},r_0)\subset \m{V}(\Delta\bSigma^{-1/2})$, and therefore the volumes must satisfy
		\begin{align}
		V_k r_0^K\leq \Vol\left(\m{V}(\Delta\bSigma^{-1/2})\right)=\frac{\Delta^K}{|\bSigma|^{1/2}}.
		\end{align}
		Now recalling that $r_0\geq Q^{-1}\left(\frac{\epsilon}{2}\right)$ by Proposition~\ref{prop:rpack} and rearranging terms, gives the result.
	\end{proof}
	\label{prop:detbound}
\end{proposition}

Proposition \ref{prop:detbound} gives us an upper bound on $\frac{|\bSigma|^{1/2K}}{\Delta}$ in terms of the ``informed'' error probability $\epsilon$. We will also need a lower bound for this ratio. A-priori, we cannot find such a bound without additional structural assumptions on the lattice spanned by $\Delta \bSigma^{-1/2}$: indeed, if the Voronoi cell has large Gaussian measure, its volume must also be large. However, the opposite isn't true, for a given Guassian measure there are convex, symmetric bodies with arbitrarily large volume. For example, think about the set $\m{S}=\{\bx\in\RR^K \ : \ |x_1|<a \}$ for some $a$. The volume of $\m{S}$ is unbounded, but $\mu(\m{S})=1-2Q(a)$, where $\mu$ is the standard Gaussian measure.

\begin{proposition}
	\label{prop:lb_with_psi}
	If $\bSigma$ satisfies assumptions $A1^*$ and $A4$, then
	\begin{equation}
	\frac{\abs{\bSigma}^{1/2k}}{\Delta} \ge  \frac{ \rho_{\text{pack}}}{K^{3/2}\cdot V_K^{1/K} Q^{-1}\left(\frac{\epsilon}{2K}\right) }.
	\end{equation}
\end{proposition}
\begin{proof}
	Recall that $d_{\text{min}}=2r_0$ is the first successive minima of the lattice $\Lambda\left(\Delta\bSigma^{-1/2}\right)$. Consider its dual lattice $\Lambda^*$ spanned by the generating matrix $\Delta^{-1}\bSigma^{1/2}$, and let $\bB=[\bb_1|\cdots|\bb_K]$ be a Korkin-Zolotarev basis of $\Lambda^*$~\cite{lls90}. In particular, there exists a matrix $\bA=[\ba_1|\cdots|\ba_K]\in\SLk$ such that $\bb_k=\Delta^{-1}\bSigma^{1/2}\ba_k$ for all $k\in[K]$. Now, from~\cite[Proposition 3.2]{lls90}, we have that for all $k\in[K]$
	\begin{align}
	\|\bb_k\|^2\cdot d_{\text{min}} \left(\Lambda\left(\Delta\bSigma^{-1/2}\right)\right)^2\leq  K^3.\label{eq:KZtrans}
	\end{align}
	Note that by definition
	\begin{align}
	\max_{k\in[K]}\|\bb_k\|^2&=\frac{1}{\Delta^2}\max_{k\in [K]}\ba_k^T\bSigma\ba_k\nonumber\\
	&\geq\frac{1}{\Delta^2} \min_{\bar{\bA}\in\SLk}\max_{k=1,\ldots,K} \bar{\ba}_k^T\bSigma\bar{\ba}_k\nonumber\\
	&=\frac{\sigma^2_K(\bSigma^{1/2})}{\Delta^2}.\label{eq:KZbound}
	\end{align}
	Combining~\eqref{eq:KZtrans} and~\eqref{eq:KZbound} yields
	\begin{align}
	r_0^2\leq K^3\frac{\Delta^2}{4\sigma^2_K(\bSigma^{1/2})}.
	\end{align}
	By~\eqref{eq:IFPeBounds}, we have that
	\begin{align}
	\frac{\Delta^2}{4\sigma^2_K(\bSigma^{1/2})}\leq \left[Q^{-1}\left(\frac{\epsilon}{2K}\right)\right]^2,
	\end{align}
	such that $r_0\leq K^{3/2} \cdot Q^{-1}\left(\frac{\epsilon}{2K}\right)$.
	Now the claim follows by recalling that
	\[
	r_0 \ge  \rho_{\text{pack}}\cdot r_{\text{eff}} = \frac{\rho_{\text{pack}}}{ V_K^{1/k}}\cdot \frac{\Delta}{\abs{\bSigma}^{1/2k}} \,.
	\]
	and rearranging terms.
\end{proof}

\begin{remark}[On the loss of integer-forcing decoding with respect to MAP decoding]
	Note that we can also use~\cite[Proposition 3.2]{lls90} in order to bound the error probability of the integer forcing decoder in terms of that of the MAP decoder. In particular,~\eqref{eq:KZtrans} reads
	\begin{align}
	\|\bb_k\|^2 \leq \frac{K^3}{d_{\text{min}} \left(\Lambda\left(\Delta\bSigma^{-1/2}\right)\right)^2}.
	\end{align}
	Now, note that in the proof of Proposition~\ref{prop:rpack} we have actually shown that $d_{\text{min}} \left(\Lambda\left(\Delta\bSigma^{-1/2}\right)\right)\geq 2Q^{-1}\left(\frac{\Pr(g_{\text{MAP}}(\bX^*)\neq \bX)}{2}\right)$. Combining this with~\eqref{eq:KZbound}, we obtain
	\begin{align}
	\sigma^2_K(\bSigma^{1/2})\leq \frac{K^3\Delta^2}{4\left(Q^{-1}\left(\frac{\Pr(g_{\text{MAP}}(\bX^*)\neq \bX)}{2}\right)\right) ^2}.
	\end{align}
	Now, applying~\eqref{eq:IFPeBounds}, we obtain that
	\begin{align}
	\Pr(\hat{\bX}_{\text{IF}}\neq\bX)
	&\leq 2K\cdot Q\left(\frac{\Delta}{2\sigma_K(\bSigma^{1/2})}\right)\nonumber\\
	&\leq 2K\cdot Q\left(\frac{Q^{-1}\left(\frac{\Pr(g_{\text{MAP}}(\bX^*)\neq \bX)}{2}\right)}{K^{3/2}}\right).
	\end{align}
	\label{rem:ifloss}
\end{remark}

Next, we provide a lower bound on $\lambda_{\text{min}}(\bSigma^{1/2})$ (smallest eigenvalue) in terms of the probability that $\bX$ avoids $\CUBE$. 
\begin{proposition}
	\label{prop:tau_min_lowerbnd}
	If $\bSigma$ satisfies assumption $A3$, then
	\begin{equation}
	\lambda_{\text{min}}(\bSigma^{1/2}) \ge \frac{\abs{\bSigma}^{1/2}}{\Delta^{K-1}} \cdot \left(\frac{2Q^{-1}\left(\frac{P}{2}\right)}{\sqrt{K}}\right)^{K-1} \,,
	\end{equation}
\end{proposition}
\begin{proof}
	Recall that $\CUBE \subset \m{B}\left(\mathbf{0},\frac{\Delta\sqrt{K}}{2}\right)$, hence
	\begin{align*}
		1-P 
		&\le \Pr(\bZ\in\bSigma^{-1/2}\CUBE)  \\
		&\le \Pr\left( \bZ\in \frac{\Delta\sqrt{K}}{2}\bSigma^{-1/2}\m{B}(\mathbf{0},1) \right)\,,
	\end{align*}
	where $\bZ = \bSigma^{-1/2}\bX \sim \m{N}(\mathbf{0},\bI_K))$. By Proposition~\ref{prop:contained_ball}, the ellipsoid $\frac{\Delta\sqrt{K}}{2}\bSigma^{-1/2}\m{B}(\mathbf{0},1)$ contains a ball of radius $r \ge Q^{-1}\left(\frac{P}{2}\right)$, which implies that all of its primary axes $\lambda_1(\bSigma^{1/2})\geq\cdots\geq\lambda_K(\bSigma^{1/2})=\lambda_{\text{min}}(\bSigma^{1/2})$ satisfy
	\[
	Q^{-1}\left(\frac{P}{2}\right) \le r \le \frac{\Delta\sqrt{K}}{2} \cdot \frac{1}{\lambda_i(\bSigma^{1/2})} \,.
	\] 
	Hence,
	\begin{align*}
		\lambda_{\text{min}}(\bSigma^{1/2}) 
		&= \frac{\abs{\bSigma}^{1/2}}{\lambda_1(\bSigma^{1/2})\cdot \ldots \cdot \lambda_{K-1}(\bSigma^{1/2})} \\
		&\ge \frac{\abs{\bSigma}^{1/2}}{\Delta^{K-1}} \cdot \left(\frac{2Q^{-1}\left(\frac{P}{2}\right)}{\sqrt{K}}\right)^{K-1} \,,
	\end{align*}
	as claimed.
\end{proof}

Combining Propositions \ref{prop:lb_with_psi} and \ref{prop:tau_min_lowerbnd}, we readily get the following bound:
\begin{proposition}
	\label{prop:chi1_bound}
	If $\bSigma$ satisfies assumptions $A1^*$, $A3$ and $A4$, then
	\begin{equation}
	\frac{\Delta}{\lambda_{\text{min}}(\bSigma^{1/2})} \le \frac{K^{2K-\frac{1}{2}}\cdot V_K \cdot \left(Q^{-1}\left(\frac{\epsilon}{2K}\right)\right)^K  }{2^K\cdot\rho_{\text{pack}}^K\cdot \left(Q^{-1}\left(\frac{P}{2}\right)\right)^{K-1} } \,.
	\end{equation}
\end{proposition}

\subsection{Performance Guarantees for the Algorithm}

We start by rewriting the bound of Theorem~\ref{thm:cov_est_ugly_bnd} in terms of quantities that are well-controlled across iterations, as we move from $\bSigma^{(m)}$ to $\bSigma^{(m+1)}$: $K$, $\epsilon$, $\Pr(\bX\notin\CUBE)$, and $\rho_{\text{pack}}$, i.e., eliminate the dependence on $\tau_{\text{min}}$, as well as eliminate the explicit dependence on the ratio $\Delta/\abs{\bSigma}^{1/2K}$.

\begin{lemma}
	\label{lem:est_guarantees_trackable_quantities}
	Suppose $\Pr(\bX\notin\CUBE) \le P$, and denote $\kappa_{\epsilon} \triangleq \frac{ Q^{-1}\left(\frac{\epsilon}{2}\right)}{\sqrt{K}}$, $\alpha(P,K)$ the lower bound from Lemma~\ref{lemma:truncvar},
	\[
	\chi_1(\epsilon;P,K) \triangleq \frac{K^{2K-\frac{1}{2}}\cdot V_K \cdot \left(Q^{-1}\left(\frac{\epsilon}{2K}\right)\right)^K  }{2^K\cdot\rho_{\text{pack}}^K\cdot \left(Q^{-1}\left(\frac{P}{2}\right)\right)^{K-1} } \,,
	\]
	the upper bound on $\frac{\Delta}{\lambda_{\text{min}}(\Sigma^{1/2})}$ from Proposition~\ref{prop:chi1_bound}, and 
	\[
	\chi_2(\epsilon;P,K)\triangleq V_K^{1/K}\cdot Q^{-1}\left(\frac{\epsilon}{2}\right) \cdot \frac{1}{\chi_1(\epsilon;P,K)} \,.
	\]
	Suppose that $\kappa_{\epsilon} > 1$, and that the distance parameter $d$ satisfies
	\[
	\eta \triangleq \frac{d}{2\sqrt{K} \cdot \frac{\Delta}{\chi_1(\epsilon;P,K)} \cdot \kappa_{\epsilon}} \in \left(0,1-\frac{1}{\kappa_{\epsilon}}\right) \,.
	\] 
	Then whenever 
{\small
	\begin{align*}
n > c_1\cdot \max \left( \frac{K}{\alpha(P;K)\cdot\chi_1(\epsilon;P,K)}\cdot\log K , \left( \frac{K}{\alpha(P;K)\cdot\chi_1(\epsilon;P,K)} \right)^2 \right)\,,	
\end{align*}
}
we have that with probability at least $1-p_{\text{est-err}}$ the estimator $\bSigmaV$ returned by the algorithm satisfies
	\begin{equation}
	0.999 \cdot \E{\bX\bX^T|\bX\in\CUBE} \preceq \bSigmaV \preceq 1.001\cdot \E{\bX\bX^T|\bX\in\CUBE} \,.
	\end{equation}
	Here,
	\begin{equation}
	p_{\text{est-err}} = p_{\text{false-positive}} + p_{\text{many-escapees}} + p_{\text{sample-est}} + p_{\text{miss-cover}}\,,
	\end{equation}
	where 
	\begin{align*}
	p_{\text{false-positive}} &\triangleq ne^{-\frac{K}{2}f((1-\eta)\kappa_{\epsilon})} \,, \\
	p_{\text{many-escapees}} &\triangleq e^{-c_2\cdot \sqrt{n}} \,, \\
	p_{\text{sample-est}} &\triangleq e^{-c_3 \frac{\alpha(P;K)}{K\cdot\left(\chi_1(\epsilon;P,K)\right)^2} \cdot n} + e^{-\frac{1}{2}(1-P)^2n} \,, 
	\end{align*}
	{\small
		\begin{align*}
		p_{\text{miss-cover}} \triangleq \exp \Biggl[& -\left(c_4\cdot \frac{ \chi_2(\epsilon;P,K) \cdot \kappa_{\epsilon}\cdot \eta}{ \chi_1(\epsilon;P,K) + 1}\right)^K e^{\frac{1}{2}\log(n) - \sqrt{\frac{1}{2}K\log(n)} - \frac{1}{2}K} \\
		&+K\log\left(\sqrt{K}+\frac{2\cdot \chi_1(\epsilon;P,K)}{\kappa_{\epsilon}\cdot \eta}\right)  \Biggr] \,,
	\end{align*}
}
	and $c_1,c_2,c_3,c_4$ are positive numerical constants.
\end{lemma}
\begin{proof}
	Recall that by Proposition~\ref{prop:chi1_bound}, we have 
	\[
	\frac{\Delta}{\lambda_{\text{min}}(\bSigma^{1/2})} \le \chi_1(\epsilon;P,K) \,,
	\]
	hence we can use Theorem \ref{thm:cov_est_ugly_bnd} with $\tau_{\text{min}}=\frac{\Delta}{\chi_1(\epsilon;P,K)}$. 
	We now need to show that the expressions above indeed bound the probabilities of their corresponding events from Theorem~\ref{thm:cov_est_ugly_bnd}. For the first two terms, there is no work to be done.
	
	Plugging in this choice of $\tau_{\text{min}}$, and choosing $\beta$ to be a small constant, we obtain $p_{\text{sample-est}}$. Since we also require that $n = \Omega \left( \left( \frac{K}{\alpha(P;K)\cdot\chi_1(\epsilon;P,K)} \right)^2 \right)$, by choosing $c_1$ large enough we can ensure that $\overline{\beta} \le 0.001$. 
	
	As for $\Pr(\m{E}_{\text{false-positive}})$, plugging our $\tau_{\text{min}}$ into the bound gives
	{\small
	\begin{align*}
		\Pr \left(\m{E}_{\text{miss-cover}}\right) 
		\le 
		\exp \Biggl[ &-\left(\frac{ \frac{1}{\chi_1(\epsilon;P,K)}\cdot \frac{\Delta}{\abs{\bSigma}^{1/2K}} \cdot \kappa_{\epsilon}\cdot \eta}{\sqrt{\frac{\pi}{2}} \left(2\cdot \chi_1(\epsilon;P,K) + 4\right)}\right)^K \\
		&\cdot e^{\frac{1}{2}\log(n) - \sqrt{\frac{1}{2}K\log(n)} - \frac{1}{2}K} \\
		&+K\log\left(\sqrt{K}+\frac{2\cdot \chi_1(\epsilon;P,K)}{\kappa_{\epsilon}\cdot \eta}\right)  \Biggr] \,. 
	\end{align*}
}
	By Proposition~\ref{prop:upper_bound_no_psi},
	\[
	\frac{\Delta}{|\bSigma|^{1/2K}}\geq V_K^{1/K}\cdot Q^{-1}\left(\frac{\epsilon}{2}\right)\,.
	\]  
	Plugging this estimate and discarding explicit numerical constants, we get $p_{\text{miss-cover}}$ above.
\end{proof}

The dependence of the bound on $n$ is somewhat subtle: as $n$ increases, the probability of the events $\m{E}_{\text{many-escapees}}$, $\m{E}_{\text{sample-est}}$ and $\m{E}_{\text{miss-cover}}$ clearly decreases. However, $\Pr\left(\m{E}_{\text{false-positive}}\right)$ increases with $n$, and moreover, according to our bound, it scales roughly like $n\epsilon^{(1-\eta)^2}$, so we require that $n \ll \epsilon^{-(1-\eta)^2}$. Since $\epsilon$ depends intricately on $\Delta,\bSigma$, it is not immediately clear that there even is a scaling of $n$ with respect to $\epsilon$ so that $p_{\text{est-err}}$ can be made very small. We now show that this is indeed the case.
%

We start by simplifying our bounds in the regime where $\epsilon$ is very small.

\begin{corollary}
	\label{cor:est_asymptotics}
	We consider the asymptotics of $p_{\text{est-err}}$ as $\epsilon\to 0$ and 
	\[
	\revAdd{n = n(\epsilon) \asymp \epsilon^{-\zeta}} \,,
	\]
	for an exponent $\zeta < (1-\eta)^2$, with $K$ and $\eta$ fixed. Then $p_{\text{est-err}} \to 0$, and moreover, it is dominated by $p_{\text{false-positive}}$,
	\[
	p_{\text{est-err}} \sim p_{\text{false-positive}} = ne^{-\frac{K}{2}f((1-\eta)\kappa_{\epsilon})} \,.
	\]
\end{corollary}
\begin{proof}
	For small $\epsilon$, we may roughly approximate $Q^{-1}(\epsilon) \sim \sqrt{ 2 \log\frac{1}{\epsilon}}$.  Hence 
	\[
	\chi_1(\epsilon;P,K) \sim \left( \log\frac{1}{\epsilon} \right)^{\frac{K}{2}} \,.
	\]
	Due to the scaling $n\sim\epsilon^{-\Omega(1)}$, the dominating terms in the exponents of $p_{\text{many-escapees}}$, $p_{\text{sample-est}}$ and $p_{\text{miss-cover}}$
	all look like 
	\[
	-\left( \log\frac{1}{\epsilon} \right)^{-\text{poly}(K)} \cdot \epsilon^{-\Omega(1)} \,,
	\]
	(in $p_{\text{miss-cover}}$ there is another negligible positive term $\sim \log\log\frac{1}{\epsilon}$). On the other hand, estimating $f(\alpha)=\alpha^2(1-o(1))$ as $\alpha\to\infty$,
	\[
	p_{\text{false-positive}} \sim n\cdot e^{-(1-\eta)^2\frac{1}{2}(Q^{-1}(\epsilon))^2} \sim n \cdot \epsilon^{(1-\eta)^2} \sim \epsilon^{(1-\eta)^2-\zeta}\,, 
	\]
	which decays to $0$ much slower than the other error terms.
\end{proof}

\vspace{6mm}
We are ready to prove Theorem~\ref{thm:success_estimator}, which bounds the success probability of the (modified) main algorithm, combining the results of both Theorems \ref{thm:genie} and \ref{thm:cov_est_ugly_bnd}.

\begin{proof}[Proof of Theorem~\ref{thm:success_estimator}]
	Choose \revAdd{$\epsilon$ small enough} to ensure that $Q^{-1}(\epsilon)>6\sqrt{K}$. Hence, by Theorem~\ref{thm:genie}, it suffices to make sure that at every iteration $t=1,\ldots,M$, the estimator $\bSigmaV^{(t)}$ returned by $\mathrm{EstimateTruncatedCovariance}$ satisfies 
	{\small \begin{align*}
		0.999 \cdot &\E{\bV^{(t)}(\bV^{(t)})^T|\bV^{(t)}\in\CUBE} \preceq \bSigmaV^{(t)} \\
		&\preceq 1.001 \cdot \E{\bV^{(t)}(\bV^{(t)})^T|\bV^{(t)}\in\CUBE} \,.
	\end{align*}}
	Note that if $\epsilon$ is small enough, Corollary \ref{cor:dynamics} implies that under this event, $P_t \le P$ for all $t$ (it suffices that the fixed point of the dynamic is smaller than $P$). Taking $\epsilon$ to be small enough, so that $0.01 < 1-\frac{1}{\kappa_{\epsilon}}$, we may apply Lemma~\ref{lem:est_guarantees_trackable_quantities} (with $\eta=0.01$) to bound the probability of failure for a single call to $\mathrm{EstimateTruncatedCovariance}$. Observe that since we make sure to use \emph{new} points at each iteration, each call to the estimation algorithm sees $n/m$ points sampled i.i.d from the truncated Gaussian distribution (with covariance $\bSigma^{(t)}$). When $\epsilon$ is sufficiently small, using Corollary~\ref{cor:est_asymptotics}, we may bound the failure probability at each iteration by $2\cdot p_{\text{false-positive}} = 2\cdot \frac{n}{m}\cdot e^{-\frac{K}{2}f(0.99\cdot \kappa_{\epsilon})}$ (we could, of course, replace $2$ by any other number $>1$). We conclude by taking the union bound over all $m$ failure events.
\end{proof}

Theorem~\ref{thm:success_estimator} shows that after the run of our algorithm, we are guaranteed with high probability to find a matrix $\bA\in\SLk$ such that $\Pr(\bA\bX\notin\CUBE)$ is small, where $\bX\sim \m{N}(\mathbf{0},\bSigma)$ is a fresh sample, independent of $\{\bX_1,\ldots,\bX_n\}$ that were used for the computation of $\bA$. In order to provide guarantees on recovering the sample $\bX_1,\ldots,\bX_n$, what we actually need to bound are the probabilities $\Pr(\bA\bX_i \notin\CUBE)$. This is \emph{not} the same thing, since here the matrix $\bA$ obviously depends on the sample points. Hence, the following result is needed.

\begin{lemma}
	\label{lem:sample_recovery}
	Suppose that our algorithms produces with probability $\ge 1-p_e$ a matrix $\bA=\bA(\bX_1^*,\ldots,\bX_n^*) \in \SLk$ such that 
	\[
	\Pr(\bA\bX\notin\CUBE) \le \epsilon' \,,
	\]
	where $\bX\sim \m{N}(\mathbf{0},\bSigma)$ is a new sample point. Then
	{\small
	\begin{equation}
	\Pr \left(\exists i\in[n] \,:\, \bA\bX_i \notin\CUBE \right) \le p_e + n\Pr\left(\norm{\bZ}>Q^{-1}(\epsilon'/2)\right) \,,
	\end{equation}
}
	where $\bZ\sim\m{N}(\mathbf{0},\bI_K)$. Assuming that $Q^{-1}(\epsilon'/2)>\sqrt{K}$, we may invoke Proposition \ref{prop:poltyrev} to further bound this by $p_e + n e^{-\frac{K}{2}f \left(\frac{Q^{-1}(\epsilon'/2)}{\sqrt{K}}\right)}$.
\end{lemma}
\begin{proof}
	For any $\bB\in\SLk$, denote the set $\m{S}_\bB = \bSigma^{-1/2}\bB^{-1}\CUBE$. Clearly,
	\[
	\Pr(\bB\bX\notin\CUBE) = 1-\mu(\m{S}_\bB)\,,
	\]
	where $\mu$ is the standard Gaussian measure. Denote the set
	\[
	\m{A} = \left\{ \bB \,:\,\mu(\m{S}_\bB) \ge 1-\epsilon' \right\} \,.
	\]
	Clearly,
	{\footnotesize
	\begin{align*}
	\Pr &\left(\exists i\in[n] \,:\, \bA\bX_i \notin\CUBE \right) \\
	&\le \Pr(\bA \notin \m{A}) + 	\Pr \left(\exists i\in[n] \,:\, \bA\bX_i \notin\CUBE | \bA\in \m{A} \right)\Pr(\bA\in\m{A}) \\
	&\le \Pr(\bA\notin\m{A}) + \Pr \left(\exists i\in[n],\exists\bB\in\m{A} \,:\, \bB\bX_i \notin\CUBE \right) \\
	&= p_e + n \Pr(\exists \bB\in\m{A}\,:\,\bB\bX\notin\CUBE) \\
	&= p_e + n \left[ 1- \mu\left( \bigcap_{\bB\in\m{A}} \m{S}_\bB \right) \right] \,.
	\end{align*} 
}
	Now, the sets $\m{S}_\bB$ are all convex, symmetric with $\mu(\m{S}_\bB) \ge 1-\epsilon'$. Hence, by Proposition~\ref{prop:contained_ball}, they all contain a centered ball of radius at least $r=Q^{-1}(\epsilon'/2)$. Hence, we can bound 
	\[
	1- \mu\left( \bigcap_{\bB\in\m{A}} \m{S}_\bB \right) \le 1-\mu\left(\m{B}(\mathbf{0},r)\right) = \Pr(\norm{\bZ}>r)\,.
	\]
	Assuming that $Q^{-1}(\epsilon'/2) > \sqrt{K}$, we may now use Proposition \ref{prop:poltyrev} to obtain the claimed bound. 
\end{proof}

Theorem~\ref{thm:success_recovery} is now a simple corollary.

\begin{proof}[Proof of Theorem~\ref{thm:success_recovery}]
	This follows immediately from Theorem \ref{thm:success_estimator} and Lemma \ref{lem:sample_recovery}.
\end{proof}

\section{Discussion}
\label{sec:summary}

We have presented a blind iterative algorithm for recovering $n$ i.i.d. $K$-dimensional Gaussian vectors from their modulo wrapped versions, whose computational complexity is $\m{O}\left(n^2\log{K}+ n\mathrm{poly}(K)\right)$ (provided that the shortest basis problem is approximated using the LLL algorithm; if the shortest basis problem is solved exactly, the runtime becomes $\m{O}\left(n^2\log{K}+ n\mathrm{poly}(K)+\left(\frac{5}{4}\right)^{K^3/4}\right)$). Our analytic results show that when the informed benchmark algorithm achieves very low error probability, and $n$ is very large, our algorithm achieves essentially the same error probability. We have also performed numerical experiments, which indicate that the algorithm performs almost as well as the informed benchmark algorithm in many scenarios of practical interest, even when $n$ is not very large and the informed error probability is not very small.

Each iteration in the proposed algorithm consists of two steps: estimating the covariance of a truncated Gaussian vector based on the modulo measurements, and then using those estimates in order to improve the currently used integer matrix. We implement the covariance estimation task by identifying the points in our sample that were not wrapped, and then computing their empirical covariance matrix.
Our procedure for identifying the unwrapped points is the source for the quadratic runtime in $n$. Furthermore, for this procedure to provably work, a certain typical set of the Gaussian distribution needs to be covered by the sample points. Thus, this step in the algorithm is also the reason for the relatively large sample complexity predicted by our analysis.\footnote{We believe that even if only a small portion of the unwrapped points, close to the origin, is identified, the main algorithm still works well. However, we were not able to prove this.} We suspect that one can find more efficient procedures for identifying the unwrapped points, which will significantly improve the proposed algorithm in terms of runtime as well as error probability. Our efforts in this direction, however, were not fruitful, and investigating this possibility further is left for future work.

\revAdd{Our analysis only dealt with recovering Gaussian random vectors from their wrapped measurements. An interesting question for future research is to what extent our findings carry over to other non-Gaussian setups. As a positive indication that our algorithm may perform well also for other distributions, we note in passing that we have repeated the numerical experiments from Section~\ref{sec:sims} with random vectors uniformly distributed over a convex set, such that the covariance matrix of the obtained random vector is the same as in the Gaussian setup, and the results were similar to those obtained for the Gaussian case.}

We conclude the paper with mentioning another potential algorithm for recovering $\{\bX_1,\ldots,\bX_n\}$ from $\{\bX^*_1,\ldots,\bX^*_n\}$. This algorithm consists of a brute-force search over the set of all feasible integer matrices. We do not bring the full analysis here, but this algorithm can be shown to achieve error probability $n\epsilon^{\xi}$, for some $\xi<1$, where $\epsilon$ is the per-sample error probability of the informed decoder, with sample complexity polynomial in $K$, but runtime $\mathrm{O}\left(ne^{\m{O}(K)}\right)$. Thus, this algorithm is not very attractive unless $K$ is very small.

Let $0<\gamma<\Delta/2$ be a design parameter. The algorithm is as follows:
\begin{enumerate}
	\item Set 
	\begin{align}
	\m{A}\triangleq\left\{\ba\in\ZZ^K \ : \ \|\ba\|\leq\left(\frac{|\bSigma|^{1/K}}{\tau_{\text{min}}}\right)^{K/2}\right\}
	\end{align}
	\item $\forall\ba\in\m{A}$ compute 
	\begin{align}
	V_\ba\triangleq\frac{1}{n}\sum_{j=1}^n\left([\ba^T\bX^*_j]^*\right)^2.
	\end{align}
	\item Construct the list
	\begin{align}
	\m{L}=\left\{\ba\in\m{A} \ : \ \sqrt{V_\ba}<\gamma \right\}.
	\end{align}
	\item If $\m{L}$ spans a subspace of rank $K$, arbitrarily choose $K$ independent vectors $\ba_1,\ldots,\ba_K$, set $\bA=[\ba_1|\cdots|\ba_K]^T$, and compute
	\begin{align}
	\hat{\bX}_i=\bA^{-1}\left([\bA\bX^*_j]^*\right), \ j\in[n].
	\end{align}
	Otherwise, declare error.
\end{enumerate}

The main observation is that the set $\m{A}$ must include the vectors $\ba^{\text{opt}}_1,\ldots,\ba^{\text{opt}}_K$, where $\bA^{\text{opt}}=[\ba^{\text{opt}}_1|\cdots|\ba^{\text{opt}}_K]$ is the solution to~\eqref{eq:Aopt}, as otherwise $\Pr(\bX\in\CUBE)\lessapprox \Pr(\bA^{\text{opt}}\bX\in\CUBE)$. Setting $\gamma$ slightly greater than $\frac{\Delta}{2}\frac{1}{Q^{-1}(\frac{\epsilon}{2})}$, guarantees that with high probability $\m{L}$ contains $\{\ba^{\text{opt}}_1,\ldots,\ba^{\text{opt}}_K\}$, but does not contain integer vectors with $\sqrt{\ba_k^T\bSigma\ba_k}$ much greater than $\frac{\Delta}{2}\frac{1}{Q^{-1}(\frac{\epsilon}{2})}$.

\bibliographystyle{IEEEtran}
\bibliography{Dissertation_bib}

\begin{IEEEbiographynophoto}{Elad Romanov}
	is currently studying towards the Ph.D. degree in the Hebrew Univeristy of Jerusalem, Israel, where he has also previously received his B.Sc. (mathematics and computer science) and M.Sc. (computer science) degrees.
	His research interests lie broadly in statistics, information theory and the mathematics of data science.
\end{IEEEbiographynophoto}
\begin{IEEEbiographynophoto}{Or Ordentlich} is a senior lecturer (assistant professor) in the School of Computer Science and Engineering at the Hebrew University of Jerusalem. He received the B.Sc. (cum laude), M.Sc. (summa cum laude),
	and Ph.D. degrees from Tel Aviv University, Israel, in 2010, in 2011, and
	2016, respectively, all in electrical engineering. During the years 2015-2017 he was a postdoctoral fellow in the Laboratory for Information and Decision Systems at the Massachusetts Institute of Technology (MIT), and in the Department of Electrical and Computer Engineering at Boston University. 
\end{IEEEbiographynophoto}

\end{document}